\documentclass{article}

\usepackage{fullpage,amsthm}
\usepackage{thmtools,thm-restate}
\usepackage{bbold,enumitem}
\usepackage{mathtools}
\usepackage{lineno}

\usepackage{amsfonts,graphicx,url,stfloats,subfig}
\usepackage{amssymb,amsmath,color,verbatim,bm}
\usepackage{stmaryrd,multirow}
\usepackage{environ,xspace}
\usepackage{thmtools,thm-restate,wrapfig}
\usepackage{tikz}

\usepackage{todonotes}
\presetkeys{todonotes}{inline}{}

\usetikzlibrary{arrows,automata,shapes}

\usetikzlibrary{calc}
\usetikzlibrary{decorations.pathreplacing}

\newcommand{\Paragraph}[1]{\noindent{\textbf{#1}}}

\newtheorem{theorem}{Theorem}

\newtheorem{lemma}[theorem]{Lemma}

\newtheorem{example}[theorem]{Example}
\newtheorem{definition}[theorem]{Definition}

\newtheorem{openquestion}[theorem]{Open question}

\newcommand{\PSPACE}{\textsc{PSpace}{}}
\newcommand{\EXPSPACE}{\textsc{ExpSpace}{}}
\newcommand{\NLOGSPACE}{\textsc{NLogSpace}{}}

\newcommand{\NP}{\textsc{NP}}

\newcommand{\N}{\mathbb{N}}
\newcommand{\Z}{\mathbb{Z}}
\newcommand{\R}{\mathbb{R}}
\newcommand{\Q}{\mathbb{Q}}

\newcommand{\tuple}[1]{\langle #1 \rangle}

\newcommand{\VASS}{\textsc{VASS}}
\newcommand{\VASSes}{\textbf{VASS}}

\newcommand{\buchi}{B\"{u}chi}

\newcommand{\flimavg}{\textsc{LimAvg}}

\newcommand{\vass}{\mathcal{A}}

\newcommand{\comp}{\pi}
\newcommand{\compCycle}{\pi^C}

\newcommand{\Path}{\rho}

\newcommand{\PathTemp}{\tau}

\newcommand{\gain}{\textsc{Gain}}
\newcommand{\select}{\textsc{Vals}}

\newcommand{\connect}{\alpha}
\newcommand{\cycle}{\beta}

\newcommand{\M}{\mathcal{M}}

\newcommand{\vectorA}{\vec{a}}
\newcommand{\vectorB}{\vec{c}}
\newcommand{\vectorC}{\vec{d}}

\newcommand{\matrixA}{\mathbf{A}}
\newcommand{\matrixB}{\mathbf{B}}
\newcommand{\matrixC}{\mathbf{C}}

\newcommand{\constA}{d}
\newcommand{\constB}{e}
\newcommand{\constC}{h}

\newcommand{\template}{\textsc{Tpl}}
\newcommand{\trans}{\textsc{Trans}}

\newcommand{\OneVector}[1]{\mathbf{1}_{#1}}
\newcommand{\myTwoDimVector}[2]{
    \bigl(\begin{smallmatrix}
#1 \\ #2
\end{smallmatrix} \bigr)
}

\title{Long-Run Average Behavior of Vector Addition Systems with States}

\author{Krishnendu Chatterjee \and Thomas A. Henzinger \and Jan Otop}

\begin{document}
\maketitle
\thispagestyle{plain}
\pagestyle{plain}

\begin{abstract}
	A vector addition system with states (VASS) consists of a finite set of states and counters. 
	A configuration is a state and a value for each counter; 
	a transition changes the state and each counter is incremented, decremented, or left unchanged.
	While qualitative properties such as state and configuration reachability have been studied for VASS, 
	we consider the long-run average cost of infinite computations of VASS. 
	The cost of a configuration is for each state, a linear combination of the counter values. 
	In the special case of uniform cost functions, the linear combination is the same for all states.
	The (regular) long-run emptiness problem is, given a VASS, a cost function, and a threshold value,
	if there is a (lasso-shaped) computation such that the long-run average value of the cost function does not exceed the threshold.
	For uniform cost functions, we show that the regular long-run emptiness problem is 
	(a)~decidable in polynomial time for integer-valued VASS, and 
	(b)~decidable but nonelementarily hard for natural-valued VASS (i.e., nonnegative counters).
	For general cost functions, we show that the problem is 
	(c)~NP-complete for integer-valued VASS, and 
	(d)~undecidable for natural-valued VASS.
	Our most interesting result is for (c) integer-valued VASS with general cost functions, 
	where we establish a connection between the regular long-run emptiness problem and quadratic Diophantine inequalities.
	The general (nonregular) long-run emptiness problem is equally hard as the regular problem in all cases except (c), 
	where it remains open.
\end{abstract}

\section{Introduction}\label{sec-intro}

\noindent{\em Vector Addition System with States (VASS). }
{\em Vector Addition Systems (VASs)}~\cite{KM69} are equivalent to Petri Nets,
and they provide a fundamental framework for formal analysis of parallel processes~\cite{EN94}.
The extension of VASs with a finite-state transition structure gives {\em Vector Addition Systems with States (VASS)}.
Informally, a VASS consists of a finite set of control states and transitions
between them, and a set of $k$ counters, where at every transition between the control states 
each counter is either incremented, decremented, or remains unchanged.
A {\em configuration} consists of a control state and a valuation of each counter, and 
the transitions of the VASS describe the transitions between the configurations.
Thus a VASS represents a finite description of an infinite-state transition system between the configurations.
If the counters can hold all possible integer values, then we call them integer-valued VASS;
and if the counters can hold only non-negative values, then we call them natural-valued VASS.

\smallskip\noindent{\em Modelling power.}
VASS are a fundamental model for concurrent processes~\cite{EN94}, and often used in performance 
analysis of concurrent processes~\cite{DKO13:conc-verification-vass,GM12:asynchronous-verification-TOPLAS,KKW10:dynamic-cutoff-detection,KKW12:coverability-proof-minim}.
Moreover, VASS have been used  
(a)~in analysis of parametrized systems~\cite{Bloem16},
(b)~as abstract models for programs for bounds and amortized analysis~\cite{SZV14}, 
(c)~in interactions between components of an API in component-based 
synthesis~\cite{FMWDR17:component-based-synthesis}.
Thus they provide a rich framework for a variety of problems in verification as well as program analysis.

\smallskip\noindent{\em Previous results for VASS.}
A computation (or a run) in a VASS is a sequence of configurations. 
The well-studied problems for VASS are as follows:
(a)~{\em control-state reachability} where given a set of target control states a computation is successful if 
one of the target states is reached;
(b)~{\em configuration reachability} where given a set of target configurations a computation is successful if 
one of the target configurations is reached. 
For natural-valued VASS,
(a)~the control-state reachability problem is \EXPSPACE-complete: the 
\EXPSPACE-hardness is shown in~\cite{Lipton:PN-Reachability,Esparza:PN}
and the upper bound follows from~\cite{RACKOFF1978223}; and
(b)~the configuration reachability problem is decidable~\cite{DBLP:conf/stoc/Mayr81,Kosaraju82,DBLP:journals/tcs/Lambert92,leroux2012vector}, and a recent breakthrough result shows that
the problem is non-elementary hard~\cite{DBLP:journals/corr/abs-1809-07115}.
For integer-valued VASS,
(a)~the control-state reachability problem is \NLOGSPACE-complete: the counters can be abstracted away and we have to solve 
the reachability problem for graphs; and 
(b)~the configuration reachability problem is \NP-complete: this is a folklore result obtained via reduction to linear 
Diophantine inequalities.

\smallskip\noindent{\em Long-run average property.}
The classical problems for VASS are qualitative (or Boolean) properties where each computation 
is either successful or not.
In this work we consider long-run average property that assigns a real value to each computation. 
A cost is associated to each configuration and the value of a computation is the long-run average
of the costs of the configurations of the computation.
For cost assignment to configuration we consider linear combination with natural coefficients 
of the values of the counters. 
In general the linear combination for the cost depends on the control state of the 
configuration, and in the special case of uniform cost functions the linear combination 
is the same across all states.

\smallskip\noindent{\em Motivating examples.}
We present some motivating examples for the problems we consider.
First, consider a VASS where the counters represent different queue lengths, and each queue consumes 
resource (e.g., energy) proportional to its length, however, for different queues the constant of the 
proportionality might differ. 
The computation of the long-run average resource consumption is modeled as a uniform cost function,
where the linear combination for the cost function is obtained from the constants of proportionality.
Second, consider a system that uses two different batteries, and the counters represent the charge levels.
At different states, different batteries are used, and we are interested in the long-run average 
charge of the used battery. 
This is modeled as a general cost function, where depending on the control state the cost function is 
the value of the counter representing the battery used in the state.

\smallskip\noindent{\em Our contributions.}
We consider the following decision problem:
given a VASS and a cost function decide whether there is a regular (or periodic) computation such 
that the long-run average value is at most a given threshold. 
Our main contributions are as follows:
\begin{enumerate}

	\item For uniform cost functions, we show that the problem is (a) decidable
	      in polynomial time for integer-valued VASS, and (b) decidable but
	      non-elementary hard for natural-valued VASS. In (b) we assume that the cost function depends 
	      on all counters.

	\item For general cost functions, we show that the problem is (a)
	      NP-complete for integer-valued VASS, and (b) undecidable for
	      natural-valued VASS.

\end{enumerate}
Our most interesting result is for general cost functions and
integer-valued VASS, where we establish an interesting connection
between the problem we consider and quadratic Diophantine
inequalities.
Finally, instead of regular computations, if we consider existence of an arbitrary computation, 
then all the above results hold, other than the NP-completeness result for 
general cost functions and integer-valued VASS (which remains open).

\smallskip\noindent{\em Related works.}
Long-run average behavior have been considered for probabilistic VASS~\cite{DBLP:conf/lics/BrazdilKKN15},
and other infinite-state models such as pushdown automata and games~\cite{CV17a,CV17b,AAHMKT14}.
In these works the costs are associated with transitions of the underlying finite-state representation. 
In contrast, in our work the costs depend on the counter values and thus on the configurations.
Costs based on configurations, specifically the content of the stack in pushdown automata, have been considered in~\cite{DBLP:conf/fsttcs/MichaliszynO17}.
Quantitative asymptotic bounds for polynomial-time termination in VASS have also been studied~\cite{BCKNVZ18,Ler18},
however, these works do not consider long-run average property.
Finally, a related model of automata with monitor counters with long-run average property have been considered
in~\cite{CHO16a,CHO16b}.
However, there are some crucial differences: in automata with monitor counters, the cost always depends on 
one counter, and counters are reset once the value is used. 
Moreover, the complexity results for automata with monitor counters are quite different from the results
we establish.   

The full proofs of presented theorems and lemmas have been relegated to the appendix.

\section{Preliminaries}

For a sequence $w$, we define $w[i]$ as the $(i+1)$-th element of $w$ (we start with $0$)
and $w[i,j]$ as the subsequence $w[i] w[i+1] \ldots w[j]$.
We allow $j$ to be $\infty$ for infinite sequences.
For a finite sequence $w$, we denote by $|w|$ its length; and for an
infinite sequence the length is $\infty$.

We use the same notation for vectors.
For a vector $\vec{x} \in \R^k$ (resp., $\Q^k$, $\Z^k$ or $\N^k$), we define $x[i]$ as the $i$-th component of $\vec{x}$.
We define the support of $\vec{x}$, denoted by $supp(\vec{x})$ as the set of components of $\vec{x}$ with non-zero values.
For vectors $\vec{x}, \vec{y}$ of equal dimension, we denote by $\vec{x} \cdot \vec{y}$, the dot-product of $\vec{x}$ and $\vec{y}$.

\subsection{Vector addition systems with states (VASS)}

A $k$-dimensional vector addition system with states (VASS) over $\Z$ (resp., over $\N$), referred to as $\VASS(\Z,k)$ (resp., $\VASS(\N,k)$), is
a tuple $\vass = \tuple{Q, Q_0, \delta}$, where
(1)~$Q$ is a finite set of states,
(2)~$Q_0 \subseteq Q$ is a set of initial states, and
(3)~$\delta \subseteq Q \times Q \times \Z^k$.
We denote by $\VASSes(\Z,k)$ (resp., $\VASSes(\N,k)$) the class of $k$-dimensional VASS over $\Z$ (resp., $\N$).
We often omit the dimension in VASS and write $\VASS(\Z),\VASS(\N), \VASSes(\N),\VASSes(\Z)$ if a definition or an argument is uniform w.r.t. the dimension.

\Paragraph{Configurations and computations}.
A \emph{configuration} of a $\VASS(\Z,k)$ $\vass$ is a pair from $Q \times \Z^k$, which consists of a state and a valuation of the counters.
A \emph{computation} of $\vass$ is an infinite sequence $\comp$ of configurations such that
(a)~$\comp[0] \in Q_0 \times \{\vec{0}\}$\footnote{Without loss of generality, we assume that the initial counter valuation is $\vec{0}$. We can encode any initial configuration in the VASS itself.}, and
(b)~for every $i \geq 0$, there exists $(q,q',\vec{y}) \in \delta$ such that
$\comp[i] = (q,\vec{x})$ and $\comp[i+1] = (q',\vec{x}+\vec{y})$.
A computation of a $\VASS(\N,k)$ $\vass$ is a computation $\comp$ of $\vass$ considered as a $\VASS(\Z,k)$ such that the values of all counters are natural, i.e.,
for all $i >0$ we have $\comp[i] \in Q \times \N^k$.

\Paragraph{Paths and cycles}.
A path $\Path = (q_0,q'_0,\vec{y}_0), (q_1,q_1',\vec{y}_1), \ldots$ in a $\VASS(\Z)$ (resp., $\VASS(\N)$) $\vass$ is a  (finite or infinite) sequence of transitions (from $\delta$)
such that for all $0 \leq i \leq |\Path|$ we have $q'_i = q_{i+1}$.
A finite path $\Path$ is a cycle if $\Path = (q_0,q'_0,\vec{y}_0), \ldots, (q_m,q'_m,\vec{y}_m)$ and $q_0 = q'_m$.
Every computation in a $\VASS(\Z)$ (resp., $\VASS(\N)$) defines a unique infinite path.
Conversely, every infinite path in a $\VASS(\Z)$ $\vass$ starting with $q_0 \in Q_0$ defines a computation in $\vass$.
However, if $\vass$ is a $\VASS(\N,k)$, some paths do not have corresponding computations due to non-negativity restriction posed on the counters.

\Paragraph{Regular computations}. We say that a computation $\comp$ of a $\VASS(\Z)$ (resp., $\VASS(\N)$) is \emph{regular} if it corresponds to a path which is ultimately periodic,
i.e., it is of the form $\alpha \beta^{\omega}$ where $\alpha, \beta$ are finite paths.

\subsection{Decision problems}

We present the decision problems that we study in the paper.

\Paragraph{Cost functions}.
Consider a $\VASS(\Z,k)$ (resp., $\VASS(\N,k)$) $\vass = (Q,Q_0,\delta)$.
A cost function $f$ for $\vass$ is a function $f \colon Q \times \Z^k \to \Z$, which
is linear with natural coefficients, i.e., for every $q$
there exists $\vec{a} \in \N^k$ such that $f(q, \vec{z}) = \vec{a} \cdot \vec{z}$ for all $\vec{z} \in \Z^k$.
We extend cost functions to computations as follows.
For a computation $\comp$ of $\vass$, we define $f(\comp)$ as the sequence $f(\comp[0]), f(\comp[1]), \ldots$.
Every cost function $f$ is given by a labeling $l \colon Q \to \N^k$ by $f(q, \vec{z}) = l(q) \cdot \vec{z}$.
We define the size of $f$, denoted by $|f|$,
as the size of binary representation of $l$ considered as a sequence of natural numbers of the length $|Q|\cdot k$.

\Paragraph{Uniform cost functions}. We say that a cost function $f$ is \emph{uniform}, if it is given by a constant function $l$, i.e.,
for all states $q,q'$ and $\vec{z} \in \Z^k$ we have $f(q, \vec{z}) =f(q', \vec{z})$.
Uniform cost functions are given by a single vector $\vec{a} \in \N^k$.

\Paragraph{The long-run average}.
We are interested in the long-run average of the values returned by the cost function, which is formalized as follows.
Consider an infinite sequence of real numbers $w$.
We define \[
	\flimavg(w) = \liminf_{k \to \infty} \frac{1}{k+1} \sum_{i=0}^k w[i].
\]

\begin{definition}[The average-value problems]
	Given a $\VASS(\Z,k)$ (resp., $\VASS(\N,k)$) $\vass$, a cost function $f$ for $\vass$, and a threshold $\lambda \in \Q$,
	\begin{itemize}
		\item the \textbf{average-value} problem (resp., the \textbf{regular average-value} problem) asks whether $\vass$ has a computation (resp., a regular computation) $\comp$ such that $\flimavg(f(\comp)) \leq \lambda$, and
		\item the \textbf{finite-value} problem (resp., the \textbf{regular finite-value} problem) asks whether $\vass$ has a computation (resp., a regular computation) $\comp$ such that $\flimavg(f(\comp)) < \infty$.
	\end{itemize}
\end{definition}

The following example illustrates the regular average-value problem for $\VASSes(\Z)$.

\begin{figure}
	\begin{center}
		\begin{tikzpicture}
			\tikzstyle{state}=[circle,draw, minimum size=0.6cm,inner sep=0cm,scale=0.85]

			\node[state] (A) at (0,0) {A};
			\node[state] (B) at (2,0) {B};
			\node[state] (C) at (4,0) {C};

			\draw[->,bend right] (B) to node[above] {$e_1,\myTwoDimVector{1}{0}$} (A);
			\draw[->,bend right] (A) to node[below] {$e_2,\myTwoDimVector{0}{-1}$} (B);
			\draw[->,bend left] (B) to node[above] {$e_3,\myTwoDimVector{0}{3}$} (C);
			\draw[->,bend left] (C) to node[below] {$e_4,\myTwoDimVector{-2}{0}$} (B);

			\draw[->] (2,0.5) to (B);

			\begin{scope}[xshift=7cm,yshift=0.5cm]
				\node[] at (0,0) {$A \mapsto \myTwoDimVector{4}{0}$};
				\node[] at (0,-0.5) {$B \mapsto \myTwoDimVector{1}{1}$};
				\node[] at (0,-1) {$C \mapsto \myTwoDimVector{0}{1}$};
			\end{scope}

		\end{tikzpicture}
	\end{center}
	\caption{The VASS $\vass_e$ and its labeling}
	\label{fig:example}
\end{figure}
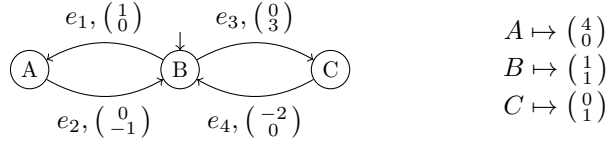

\begin{example}
	\label{ex:runningOne}
	Consider the $\VASS(\Z,2)$ $\vass_e = \tuple{\{A,B,C\}, \{B\}, \{e_1, e_2,e_3,e_4\}}$ depicted in Figure~\ref{fig:example} and
	a (non-uniform) cost function $f$ is given by the labeling from Figure~\ref{fig:example}.
	$A \mapsto \myTwoDimVector{4}{0}$,

	Consider an infinite path  $(e_1 e_2 e_3 e_4)^{\omega}$ that defines the regular computation $\comp_0$.
	This computation can be divided into blocks of 4 consecutive configurations. The $(i+1)$-th block has the following form
	\[
		\ldots (B,\myTwoDimVector{-i}{2i}) (A,\myTwoDimVector{-i+1}{2i}) (B,\myTwoDimVector{-i+1}{2i-1}) (C,\myTwoDimVector{-i+1}{2i+2})
		\ldots
	\]
	and the corresponding values of $f(\comp)$ are:
	\[
		\ldots\quad i\quad -4i+4\quad\quad i\quad\quad 2i+2\quad \ldots
	\]
	Therefore, the sum of values over each  block is $6$ and hence $\flimavg(f(\comp_0)) = \frac{3}{2}$.

	The answer to the (regular) average-value problem with any threshold is YES. Consider $j \in \N$ and a path
	$(e_1e_2)^j(e_1 e_2 e_3 e_4)^{\omega}$.
	Observe that it defines a regular computation $\comp_j$, which after the block $(e_1 e_2)^j$ coincides with $\comp_0$ with all counter values shifted by $\myTwoDimVector{-2j}{3j}$.
	Note that in each block $e_1 e_2 e_3 e_4$, the value $f$ on the vector $\myTwoDimVector{-2j}{3j}$ are $j,-8j,j,3j$.
	Therefore, $\flimavg(f(\comp_j)) = \flimavg(f(\comp_0)) + \frac{-3j}{4} = \frac{-3j-6}{4}$.
	It follows that for every threshold $\lambda \in \Q$, there exists a regular computation $\comp_j$ with $\flimavg(f(\comp_j)) \leq \lambda$.
\end{example}

\Paragraph{Organization}.
In this paper, we study the (regular) average-value problem for $\VASSes(\Z,k)$ and $\VASSes(\N,k)$.
First, we study the average-value problem for uniform cost functions (Section~\ref{s:uniform}).
Next, we consider the non-uniform case, where we focus on $\VASSes(\Z)$ (Section~\ref{s:integer}).
We start with solving the (regular) finite-value problem, and then we move to the regular average-value problem.
We show that both problems are $\NP$-complete.
Next, we discuss the non-uniform case for $\VASSes(\N)$ and we show that the regular finite-value problem is decidable (and non-elementary), while
the (regular) average-value problem is undecidable.

\section{Uniform cost functions}
\label{s:uniform}
In this section we study the average-value problem for $\VASSes(\Z)$ and $\VASSes(\N)$ for the uniform cost functions.

\subsection{Integer-valued VASS: $\VASSes(\Z)$}

Consider a $\VASS(\Z,k)$ $\vass$ and a uniform cost function $f$ for $\vass$.
This function is defined by a single vector of coefficients $\vec{a} \in \N^k$.
First, observe that we can reduce the number of the counters to one, which tracks the current cost.
This counter $c$ stores the value of $f(\pi[i])= \vec{a} \cdot \vec{z}$, where $\vec{z}$ are the values of counters.
Initially, the counter $c$ is $0 = \vec{a} \cdot \vec{0}$.
Next, in each step $i>0$, counters $\vec{z}$ are updated with some values $\vec{y}$ and we update $c$ with $\vec{a} \cdot \vec{y}$.
Note that the value of $f(\comp[i+1]) = \vec{a} \cdot (\vec{z} + \vec{y})$ equals $f(\comp[i]) + \vec{a} \cdot \vec{y}$.

Second, observe that the average-value problem for $\VASSes(\Z,1)$ and uniform cost functions is equivalent to
single-player average energy games (with no bounds on energy levels), which are solvable in polynomial time~\cite{DBLP:journals/acta/BouyerMRLL18}.
Moreover, average-energy games (with no bounds on energy levels) admit memoryless winning strategies and hence
the average-value and the regular average-value problems coincide.
In consequence we have:

\begin{theorem}
	The average-value and the regular average-value problems for $\VASSes(\Z)$ and uniform cost functions are decidable in polynomial time.
\end{theorem}

\subsection{Natural-valued VASS: $\VASSes(\N)$}

For natural-valued VASS, we cannot reduce the number of counters to one; we need to track all counters to make sure that all of them have non-negative values.
Moreover, we show that the average-value problem for (single counter) $\VASSes(\N,1)$ is in $\PSPACE$ while the average-value problem for $\VASSes(\N)$ is nonelementary.

First, the average-value problem for (single counter) $\VASSes(\N,1)$ in the uniform case is equivalent to single-player average energy games with non-negativity constraint on the energy values.
The latter problem is in $\PSPACE$ and it is $\NP$-hard~\cite{DBLP:journals/acta/BouyerMRLL18}.

In the multi-counter case, we show that the average-value problem is mutually reducible to the configuration reachability problem for VASS, which has
recently been shown nonelementary hard~\cite{DBLP:journals/corr/abs-1809-07115}.
We additionally assume that the cost function depends on all its arguments, i.e., all coefficients are non-negative. This assumption allows us show that if there is a computation of the average value below some $\lambda$, then there is one, which is lasso-shaped and the cycle in the lasso has exponential size in the size of the VASS and $\lambda$.
Therefore, we can non-deterministically pick such a cycle and check whether it is reachable from the initial configuration.

\begin{theorem}
	\label{th:decidable-average-N}
	The average-value and the regular average-value problems for $\VASSes(\N)$ and uniform cost functions with non-zero coefficients are decidable and mutually reducible to the configuration reachability problem for $\VASSes(\N)$.
\end{theorem}
\begin{proof}
	We present both directions of the reduction.\smallskip

	\noindent\emph{Hardness}.
	Consider a $\vass_0 \in \VASSes(\N,k)$ and two configurations $(q,\vec{v}), (q',\vec{v}')$.
	Without loss of generality, we assume that $\vec{v} = \vec{v}' = \vec{0}$.
	We construct a $\VASS(\N,k+1)$ $\vass_1$ build as an extension of $\vass_0$ with two additional states $q_S, q_F$ and an additional counter $c_{k+1}$.
	The state $q_S$ is the only initial state with a single outgoing transition, which goes to $q$, where only the counter $c_{k+1}$ is incremented, and other counters are unchanged.
	Then, the counter $c_{k+1}$ is not changed except for the state $q_F$.
	The state $q_F$ is a sink with an incoming transition $(q',q_F,\vec{0})$ and two loops over $q_F$.
	One loop does not change the counters and the other that decrements $c_{k+1}$ (and does not change the other counters).
	The cost function returns the sum of counters.
	Observe that every computation of the average-value $0$ has to reach $q_F$ with first $k$ counters being $0$.
	Such a computation exists  if and only if $(q',\vec{0})$ is reachable in $\vass_0$ from $(q,\vec{0})$.
	\smallskip

	\noindent\emph{Decidability}.
	We can solve the average-value problem for $\VASSes(\N)$ and uniform cost functions with non-zero coefficients
	in non-deterministic exponential time with a single query to the VASS-reachability oracle.

	Consider a $\VASS(\N,k)$ $\vass$ and assume it has a computation $\comp$ of the average-value not exceeding $\lambda$.
	Then, for every $\epsilon >0$, the computation $\comp$ contains a cyclic computation  (a cycle w.r.t. the state and the valuation of counters)
	of the average value $M = \lambda+\epsilon$.
	Let $\compCycle$ be a cyclic computation of the minimal length among computations of the average value $M$.
	Consider three sets of positions in $\compCycle$:
	\begin{enumerate}[label=(\alph*)]
		\item $N_0$, where the value of $f$ is less than $M$,
		\item $N_1$, where the value of $f$ belongs to $[M,2M]$, and
		\item $N_2$, where the value of $f$ exceeds $2M$.
	\end{enumerate}
	Note that $|N_2| \leq |N_0|$, as otherwise the average value exceeds $M$.
	Since $\compCycle$ has the minimal-length, it does not contain a subcycle computation. A subcycle computation has either the average less-or-equal $M$, and hence it is shorter than $\compCycle$, or it has the average greater than $M$, and hence it can be removed.
	Therefore, all positions in $\compCycle$ are distinct. Finally, since
	all coefficients are non-zero, the value of each counter is bounded by $M$,  and hence $|N_0|, |N_1| \leq |Q|\cdot (2M)^k$.
	In consequence, the length of $\compCycle$ is bounded by $3|Q|\cdot (2M)^k$.

	This has two consequences. First, there are finitely many such cyclic computations and hence there exists a cyclic computation of the value at most $\lambda$ and length $3|Q|\cdot (2\lambda)^k$.
	Second, to decide the average-value problem for an $\VASS(\N,k)$, we can non-deterministically generate a cyclic computation of length bounded by $3|Q|\cdot (2\lambda)^k$,
	verify that (a)~its average does not exceed $\lambda$, and (b)~it is reachable from some initial configuration.

	Observe that the above procedure implies that the average-value and the regular average-value problems coincide for $\VASSes(\N)$ under
	uniform cost functions with non-zero coefficients.
\end{proof}

We have used in the proof the fact that $f$ depends on all counters. We conjecture that this assumption can be lifted:

\begin{openquestion}
	Is the average-value problem for $\VASSes(\N)$ and uniform cost functions decidable?
	\label{open:natural}
\end{openquestion}

\section{General cost functions and $\VASSes(\Z)$}
\label{s:integer}

\newcommand{\matrixIN}{\mathbf{B}}

First, we consider (an extension of) the regular finite-value problem for $\VASSes(\Z)$ (Section~\ref{sec:finite-average}).
We show that one can decide in non-deterministic polynomial time
whether a given VASS has a regular computation
(a)~of the value $-\infty$, or
(b)~of some finite value.
To achieve this, we introduce path summarizations, which allow us to state conditions that entail (a) and respectively (b).
We show that these conditions can be checked in $\NP$.

We apply the results from Section~\ref{sec:finite-average} to solve the regular average-value problem (Section~\ref{sec:regular-value}).
We consider only VASS that have some computation of a finite value, and no computation of the value $-\infty$, i.e., the answer to (a) is NO and the answer to (b) is YES.
In other cases we can easily answer to the regular average-value problem.
If the answer to (a) is YES, then for any threshold we answer YES.
If the answers to (a) and (b) are NO, then for any threshold we answer NO.
Finally, we show $\NP$-hardness of the regular finite-value and the regular average-value problems (Section~\ref{sec:np-hardness}).

The main result of this section is the following theorem:

\begin{theorem}
	\label{th:IntegerMain}
	The regular finite-value and the regular average-value problems for $\VASSes(\Z)$ with (general) cost functions are $\NP$-complete.
\end{theorem}

We fix a $\VASS(\Z,k)$ $\vass = \tuple{Q, Q_0, \delta}$, with the set of states $Q$ and the set of transitions $\delta$, and a cost function $f$, which we refer to throughout this section.
We exclude the complexity statements, where the asymptotic behavior applies to all $\vass$ and $f$.

\subsection{The finite-value problem}
\label{sec:finite-average}

For a path $\Path$, we define characteristics $\gain(\Path), \select(\Path)$, which summarize the impact of $\Path$ on the values of counters ($\gain$) and
the value of the (partial) average of costs ($\select$).

Let $\Path$ be of the form $(q_1, q_2,\vec{y}_1) \ldots (q_m, q_{m+1}, \vec{y}_m)$.
We define $\gain(\Path)$ as the sum of updates along $\Path$, i.e., $\gain(\Path) = \sum_{i=1}^m \vec{y}_i$.
The vector $\gain(\Path)$ is the update of counters upon the whole path $\Path$.
Observe that
\[
	\gain(\Path_1 \Path_2) = \gain(\Path_1) + \gain(\Path_2).
\]

Let $l \colon Q \to \N^k$ be the function representing $f$, i.e., for every $\vec{z} \in \Z^k$ we have
$f(q, \vec{z}) = l(q) \cdot \vec{z}$.
We define $\select(\Path)$ as the sum of vectors $l(q_i)$ along $\Path$, i.e., $\select(\Path) = \sum_{i=1}^m l(q_i)$.
Note that we exclude the last state $q_{m+1}$, and hence we have
\[
	\select(\Path_1 \Path_2) = \select(\Path_1) + \select(\Path_2).
\]
The vector $\select(\Path)$ describes the coefficients with which each counter contributes to the average along the path $\Path$.

Consider a regular computation $\comp$ and a path $\Path_1 (\Path_2)^{\omega}$ that corresponds to it.
Let $\comp'$ be the (regular) computation obtained from $\comp$ by contracting $|\Path_2|$ to a single transition $(q_f,q_f, \gain(\Path_2))$, which is a loop over a fresh state $q_f$.
The counters over this transition are updated by $\gain(\Path_2)$ and
the cost function in $q_f$ is defined as $f(q_f, \vec{z}) = \frac{1}{|\Path_2|}\select(\Path_2) \cdot \vec{z}$.

The values of $\comp$ and $\comp'$ may be different, but they differ only by some finite value.
Indeed, the difference over a single iteration of $\Path_2$ updates and computation of the partial averages are interleaved along $\Path_2$, while in $(q_f,q_f, \gain(\Path_2))$  we first compute the partial average and then update the counters.
Therefore, that difference is a finite value $N$ that does not depend on the initial values of counters.
It follows that the difference between the average values of (the computations corresponding to) $\Path_2$ and $(q_f,q_f, \gain(\Path_2))$ is bounded by $N$.

Finally, observe that the value of $\comp'$ can be easily estimated.
Observe that the partial average of the first $n+1$ values of $f$ in $\comp'$
equals
\[
	\frac{1}{n+1}\Big(\vec{0}\cdot \select(\Path_2) +\gain(\Path_2)\cdot \select(\Path_2) + \ldots +
	n\gain(\Path_2)\cdot \select(\Path_2)\Big) = \frac{n}{2}\gain(\Path_2)\cdot \select(\Path_2).
\]
In consequence, we have the following:

\begin{restatable}{lemma}{CycleCharacteristics}
	\label{l:cycles}
	Let $\comp$ be a regular computation corresponding to a path $\Path_1 (\Path_2)^{\omega}$. Then, one of the following holds:
	\begin{enumerate}
		\item $\gain(\Path_2) \cdot \select(\Path_2) < 0$ and $\flimavg(f(\comp)) = -\infty$, or
		\item $\gain(\Path_2) \cdot \select(\Path_2) = 0$ and $\flimavg(f(\comp))$ is finite, or
		\item $\gain(\Path_2) \cdot \select(\Path_2) > 0$ and $\flimavg(f(\comp)) = \infty$.
	\end{enumerate}
\end{restatable}

\begin{example}
	\label{ex:runningTwo}
	Consider the VASS $\vass_e$ and the cost function $f$ from Example~\ref{ex:runningOne}.
	We have shown that the computation defined by the path $(e_1 e_2 e_3 e_4)^{\omega}$ has finite average value.
	This can be algorithmically computed using Lemma~\ref{l:cycles}.
	We compute
	\[
		\begin{split}
			\gain(e_1 e_2 e_3 e_4) & = \myTwoDimVector{1}{0} + \myTwoDimVector{0}{-1} + \myTwoDimVector{0}{3} +\myTwoDimVector{-2}{0} = \myTwoDimVector{-1}{2} \\
			\select(e_1 e_2 e_3 e_4) & = \myTwoDimVector{1}{1} + \myTwoDimVector{4}{0} + \myTwoDimVector{1}{1} +  \myTwoDimVector{0}{1} = \myTwoDimVector{6}{3}
		\end{split}
	\]
	Therefore,
	$\gain(e_1 e_2 e_3 e_4) \cdot \select(e_1 e_2 e_3 e_4) = 0$.
\end{example}

\Paragraph{Reduction to integer quadratic programming}.
Lemma~\ref{l:cycles} reduces the regular finite-value problem to finding a cycle with $\gain(\Path) \cdot \select(\Path) < 0$ (resp., $\gain(\Path) \cdot \select(\Path) = 0$).
We show that existence of such a cycle can be stated as a polynomial-size instance of \emph{integer quadratic programming}, which can be decided in $\NP$~\cite{del2017mixed}.
An instance of integer quadratic programming consists of a symmetric matrix
$\matrixA \in \Q^{n \times n}, \vectorA \in \Q^n, \constA \in \Q, \matrixB \in \Q^{n \times m}, \vectorB \in \Q^m$ and it asks whether there exists
a vector $\vec{x} \in \Z^n$ satisfying the following system:
\[
	\begin{split}
		\vec{x}^T \matrixA \vec{x} + \vectorA \vec{x} + \constA &\leq 0 \\
		\matrixIN \vec{x}  &\leq \vectorB
	\end{split}
\]

Let $\Path$ be a cycle.
Observe that both $\gain(\Path), \select(\Path)$ depend only on the multiplicity of transitions occurring in $\Path$; the order of transitions is irrelevant.
Consider a vector $\vec{x} = (x[1], \ldots, x[m])$ of multiplicities of transitions $e_1, \ldots, e_m$ in $\Path$.
Then,
\[
	\gain(\Path) \cdot \select(\Path) = \sum_{1 \leq i,j \leq m} x[i] x[j] \gain(e_i) \cdot \select(e_j)
\]
Therefore, we have
\begin{equation}
	2 \cdot \gain(\Path) \cdot \select(\Path) = \vec{x}^T \matrixA \vec{x}
	\label{eq:matrix-balanced}
\end{equation}
for a symmetric matrix $\matrixA \in \Z^{m\times m}$ defined for all $i,j$ as \[
	\matrixA[i,j] = \gain(e_i) \cdot \select(e_j) + \gain(e_j) \cdot \select(e_i).
\]
The left hand side of \eqref{eq:matrix-balanced} is multiplied by $2$ to avoid division by $2$.
It follows that if $\gain(\Path) \cdot \select(\Path) < 0$ (resp., $\gain(\Path) \cdot \select(\Path) = 0$), then the inequality
$\vec{x}^T \matrixA \vec{x} -1 \leq 0$ (resp., $\vec{x}^T \matrixA \vec{x} \leq 0$) has a solution.

Note that the above inequality can have a solution, which does not correspond to a cycle.
We can encode with a system of linear inequalities that $\vec{x}$ corresponds to a cycle.
Consider $S \subseteq Q$, which corresponds to all states visited by $\Path$.
It suffices to ensure that
(a)~for all $s \in S$, the number of incoming transitions to $s$, which is the sum of multiplicities of transitions leading to $s$, equals the number of transitions outgoing from $s$,
(b)~for every $s \in S$, the number of outgoing transitions is greater or equal to $1$, and
(c)~for $s \in Q \setminus S$, the number of incoming and outgoing transitions is $0$.
Finally, we state that all multiplicities are non-negative.
We can encode such equations and inequalities as $\matrixIN_{S} \vec{x} \leq \vectorB_S$.
Observe that every vector of multiplicities $\vec{x} \in \Z$ satisfies $\vec{x}^T \matrixA \vec{x} -1 \leq 0$ (resp., $\vec{x}^T \matrixA \vec{x} \leq 0$) and
$\matrixIN_S \vec{x} \leq \vectorB_S$ defines a cycle $\Path$ with $\gain(\Path), \select(\Path) < 0$ (resp., $\gain(\Path), \select(\Path) =0$).
The matrices $\matrixA, \matrixIN_S$ and the vector $\vectorB_S$ are polynomial in $|\vass|$. The set $S$ can be picked non-deterministically.
In consequence, we have the following:

\begin{restatable}{lemma}{MatrixNegative}
	\label{l:MatrixNegative}
	The problem: given a $\VASS(\Z,k)$ $\vass$ and a cost function $f$,
	decide whether $\vass$ has a regular computation of the value $-\infty$ (resp., less than $+\infty$) is in $\NP$.
\end{restatable}

\begin{example}
	\label{ex:runningThree}
	Consider the VASS $\vass_e$ and the cost function $f$ from Example~\ref{ex:runningOne}.
	We consider some cycle $\Path$ starting in $B$ over all states of $\vass_e$. Let $\vec{x} \in \N^4$ be the multiplicities of transitions $e_1, e_2, e_3, e_4$ respectively in $\Path$
	Then, we have
	\[
		\matrixA = \begin{bmatrix*}[r]
			2 & 3 & 4 & -2 \\
			3 & 0 & -1 & -9 \\
			4 & -1 & 6 & 1 \\
			-2 & -9 & 1 & 0 \\
		\end{bmatrix*}
	\]
	Since $\Path$ is a cycle starting from $B$, we have $x_1 = x_2$ and $x_3 = x_4$. Therefore, we can eliminate $x_2, x_4$ and the matrix $\matrixA$ can be simplified to $\matrixA' \in \Z^{2\times 2}$
	defined as
	\[
		\matrixA' = \begin{bmatrix*}[r]
			8 & -8 \\
			-8 & 8
		\end{bmatrix*}
	\]
	where $x_1$ (resp., $x_3$) is the multiplicity of the cycle $e_1 e_2$ (resp., $e_3 e_4$).
	Let $\vec{x}' = (x_1, x_3)$.
	Note that \[
		(\vec{x}')^T \matrixA' \vec{x} = (x_1 - x_3)^2
	\]
	It follows that a cycle $\Path$ defines the computation of some finite value if and only if it contains the multiplicities of transitions $e_1, e_2, e_3, e_4$ are equal.
\end{example}

\subsection{The regular average-value problem}
\label{sec:regular-value}

\newcommand{\VAL}{\mathsf{Sum}}

In this section, we study the regular average-value problem for $\VASSes(\Z)$.
We assume that $\vass$ has a regular computation of finite value and does not have a regular computation of the value $-\infty$.
Finally, we consider the case of the threshold $\lambda = 0$.
The case of an arbitrary threshold $\lambda \in \Q$ can be easily reduced to the case $\lambda = 0$
(see~\cite{CHO16a,CHO16b} for intuitions and techniques in mean-payoff games a.k.a. limit-average games).

Consider a regular computation $\comp$ and let $\Path_1 (\Path_2)^\omega$ be a path corresponding to $\comp$.
We derive the necessary and sufficient conditions on $\Path_1, \Path_2$ to have $\flimavg(f(\comp)) \leq 0$.

Assume that $\flimavg(f(\comp)) \leq 0$.
Due to Lemma~\ref{l:cycles}, we have $\gain(\Path_2)\cdot \select(\Path_2) = 0$, which implies that every iteration of $\Path_2$
has the same average. Therefore, $\flimavg(f(\comp))$ is the average value of $\Path_2$ (starting with counter values defined by $\Path_1$).
The latter value is at most $0$ if and only if the sum of values along $\Path_2$ is at most $0$. We define this sum below.

\Paragraph{The sum of values $\VAL_{\vec{g}}(\Path)$}.
Let $l$ be the labeling defining the cost function $f$.
Consider a path $\Path$ of length $m$ such that
$\Path = (q_1,q_2, \vec{y}_1) \ldots (q_m,q_{m+1}, \vec{y}_m)$.
Let $\comp$ be the computation corresponding to $\Path$ with the initial counter values being $\vec{g}$.
Then, the value of counters at the position~$i$ in $\comp$ is $\vec{g} +  \sum_{j=1}^{i-1} \vec{y}_j$.
Therefore, the sum of values over $\Path$ starting with counter values $\vec{g} \in \Z^k$,
denoted by $\VAL_{\vec{g}}(\Path)$,  is given by the following formula:
\[
	\VAL_{\vec{g}}(\Path) = \sum_{i=1}^{m} \Big(\vec{g} +  \sum_{j=1}^{i-1} \vec{y}_j\Big)\cdot l(q_i)
\]
We have the following:

\begin{restatable}{lemma}{ConditionsAvg}
	There exists a regular computation $\pi$ with $\flimavg(f(\comp)) \leq 0$  if and only if there exist paths $\Path_1, \Path_2$ such that
	\begin{enumerate}[label=(C\arabic*),leftmargin=0.8cm,topsep=0pt,itemsep=-1ex,partopsep=1ex,parsep=1ex]
		\item $\Path_1 (\Path_2)^\omega$ is an infinite path and $\Path_2$ is a cycle,
		\item $\gain(\Path_2) \cdot \select(\Path_2) = 0$, and
		\item $\VAL_{\gain(\Path_1)}(\Path_2) \leq 0$.
	\end{enumerate}
	\label{l:conditions-avg}
\end{restatable}

The proof of Lemma~\ref{l:conditions-avg} has been relegated to the appendix.

Due to results from the previous section, we can check in $\NP$ the existence of $\Path_1, \Path_2$ satisfying (C1) and (C2).
We call a cycle $\Path_2$ \emph{balanced} if $\gain(\Path_2) \cdot \select(\Path_2) = 0$.
In the remaining part we focus on condition (C3), while keeping in mind (C1) and (C2).

\Paragraph{The plan of the proof}.
The proof has three key ingredients which are as follows: \emph{factorizations},  \emph{quadratic factor elimination}, and
\emph{the linear case}.
\begin{itemize}[leftmargin=10pt,itemindent=15pt]
	\item {\emph{Factorizations}}.
	      We show that we can consider only paths $\Path_2$ of the form
	      \[
		      \connect_0 \cycle_1^{n[1]} \connect_1 \cycle_2^{n[2]} \ldots \cycle_p^{n[p]} \connect_p,
	      \] where
	      $|\alpha_i|,|\beta_i| \leq |\vass|$ and each $\beta_i$ is a distinct simple cycle (Lemma~\ref{l:minimal-factorization}).
	      It follows that $p$ is exponentially bounded in $|\vass|$.
	      Next, we show that for $\Path_2$ in such a form we have
	      \[
		      \VAL_{\gain(\Path_1)}(\Path_2) = \vec{n}^T \matrixB  \vec{n} + \vectorB \cdot \vec{n} + \constB
	      \]
	      where $\matrixB \in \Z^{p\times p}, \vectorB \in \Z^p$ (Lemma~\ref{l:template-value}).
	      Therefore, $\VAL_{\gain(\Path_1)}(\Path_2) \leq 0$ can be presented as an instance of integer quadratic programming,
	      where the variables correspond to multiplicities of simple cycles.

	      However, there are two problems to overcome.
	      First, $\Path_2$ has to be balanced, i.e., it has to satisfy $\gain(\Path_2) \cdot \select(\Path_2) = 0$, which introduces another quadratic equation.
	      Solving a system of two quadratic equations over integers is considerably more difficult (see~\cite{del2017mixed} for references).
	      Second, $p$ is (bounded by) the number of distinct simple cycles and hence it can be exponential in $|\vass|$.
	      Therefore, $\matrixB$ and $\vectorB$ may have exponential size (of the binary representation) in $|\vass|$.
	      \smallskip

	\item {\emph{Quadratic factor elimination}}.
	      To solve these problems, we fix a sequence $\template = (\alpha_0,\beta_1, \alpha_1, \ldots, \beta_p, \alpha_p)$  and consider
	      $\matrixB_{\template}, \vectorB_{\template}, \constB_{\template}$ for $\template$.
	      We show that one of the following holds (Lemma~\ref{l:negative-is-good} and Lemma~\ref{l:negative-or-linear}):
	      \begin{itemize}[leftmargin=10pt,itemindent=15pt]
		      \item $\vec{n}^T \matrixB_{\template}  \vec{n} + \vectorB_{\template} \cdot \vec{n} + \constB_{\template} \leq 0$ has a simple solution, or
		      \item there exist $\vectorC_{\template} \in \Z^p, \constC_{\template} \in \Z$ such that for
		            all vectors $\vec{n}$, if the cycle $\Path_2 = \connect_0 \cycle_1^{n[1]} \ldots \cycle_p^{n[p]} \connect_p$ is balanced (C2),
		            then
		            \[ \vec{n}^T \matrixB_{\template}  \vec{n} = \vectorC_{\template}\cdot \vec{n} + \constC_{\template}.
		            \]
	      \end{itemize}
	      We can decide in non-deterministic polynomial time whether the first condition holds (Lemma~\ref{l:linear-NP}).
	      \smallskip

	\item {\emph{The linear case}}.
	      Assuming that the second condition holds, we reduce the problem of solving the quadratic inequality $\vec{n}^T \matrixB_{\template}  \vec{n} + \vectorB_{\template} \cdot \vec{n} + \constB_{\template} \leq 0$
	      to solving the linear inequality
	      \[
		      (\vectorB_{\template} + \vectorC_{\template}) \cdot \vec{n} + (\constB_{\template} + \constC_{\template}) \leq 0.
		      \tag{lin}
		      \label{eq:linear-short}
	      \]
	      Moreover, we can compute $\vectorC_{\template}, \constC_{\template}$ from the sequence $\template$.
	      At this point, we can solve the problem in non-deterministic exponential time.
	      Next, we argue that we do not have to compute the whole system.

	      We show that if $(\vectorB_{\template} + \vectorC_{\template}) \cdot \vec{n} +  \constB_{\template} + \constC_{\template} \leq 0$ has a solution, then
	      it has a solution for a vector $\vec{n}_0$ with $m = O(|\vass|)$ non-zero components.
	      Therefore, we can remove cycles corresponding to $0$ coefficients of $\vec{n}_0$.
	      Still, $\sum_{i=0}^{p} |\connect_i|$  can be exponential in $|\vass|$, but this operation shortens the size of the template, i.e., the value
	      $\sum_{i=0}^{p} |\connect_i| + \sum_{i=1}^p |\cycle_i|$, and hence by iterating it we get a polynomial size template, which
	      yields a polynomial-size system of inequalities. These inequalities can be solved in $\NP$.
\end{itemize}

\noindent We now present the details of each ingredient.

\subsubsection{Factorizations}

The regular finite-value problem has been solved via reduction to solving (quadratic and linear) inequalities.
In this section, we show a reduction of the regular average-value problem to linear and quadratic inequalities as well. First, we establish that we can consider only
cycles $\Path$, which have a compact representation using \emph{templates} parametrized by multiplicities of cycles.
The value $\VAL_{\vec{g}}(\Path)$ for a cycle represented by a template is given by a quadratic function in the multiplicities of cycles.

\Paragraph{Templates and multiplicities}.
A \emph{template} $\template$ is a sequence of paths $(\connect_0, \cycle_1, \connect_1, \ldots, \cycle_p,\connect_p)$ such that
all $\cycle_1, \ldots, \cycle_p$ are cycles and $\connect_0 \cycle_1 \connect_1 \ldots \cycle_p \connect_p$ is a cycle.
A template is \emph{minimal} if for all $i \in \{0,\ldots, p\}$ we have $|\alpha_i| < |Q|$ and
all $\beta_i$ are pairwise distinct simple cycles.
For every vector $\vec{n} \in \N^p$, called \emph{multiplicities}, we define
$\template(\vec{n})$ as a cycle
\[
	\connect_0 \cycle_1^{n[1]} \connect_1 \cycle_2^{n[2]} \ldots \cycle_p^{n[p]} \connect_p.
	\tag{$*$}
	\label{factors}
\]
A cycle $\Path$ has a (minimal) \emph{factorization} if there exists a (minimal) template and multiplicities $\vec{n}$ such that
$\Path = \template(\vec{n})$.

Observe that every cycle $\Path$ has a factorization such that for all $i$ we have $|\alpha_i| < |Q|$ and each $\beta_i$ is a simple cycle.
However, the sequence $\beta_i$'s can have repetitions.
The following lemma states that if a cycle $\cycle$ occurs twice in $\Path$, then we can group them together.

\begin{lemma}[Cycle grouping]
	\label{l:cycle-grouping}
	Consider $\vec{g} \in \Z^k$ and a cycle $\connect_0 \cycle \connect_1 \cycle \connect_2$. Then, one of the following holds:
	\begin{align*}
		\VAL_{\vec{g}}(\connect_0 \cycle^2 \connect_1 \connect_2) \leq \VAL_{\vec{g}}(\connect_0 \cycle \connect_1 \cycle \connect_2) \\
		\VAL_{\vec{g}}(\connect_0 \connect_1 \cycle^2 \connect_2) \leq \VAL_{\vec{g}}(\connect_0 \cycle \connect_1 \cycle \connect_2)
	\end{align*}
\end{lemma}
\begin{proof}[Proof sketch]
	To show the lemma observe that
	\[
		\begin{split}
			\VAL_{\vec{g}}(\connect_0 \cycle \connect_1  \cycle \connect_2) = \frac{1}{2} \big(&\VAL_{\vec{g}}(\connect_0 \cycle^2 \connect_1 \connect_2) + \\
			&\VAL_{\vec{g}}(\connect_0  \connect_1 \cycle^2 \connect_2)\big)
		\end{split}
	\]
\end{proof}

Careful repeated application of Lemma~\ref{l:cycle-grouping} implies that we can look for a cycle $\Path$ satisfying (C1), (C2) and (C3) among cycles
that have a minimal factorization:

\begin{lemma}
	\label{l:minimal-factorization}
	For every cycle $\Path$ and $\vec{g} \in \Z^k$, there exists a cycle $\Path'$ that has a minimal factorization~\eqref{factors}
	such that
	$\gain(\Path) = \gain(\Path')$,
	$\select(\Path) = \select(\Path')$ and
	$\VAL_{\vec{g}}(\Path) \geq \VAL_{\vec{g}}(\Path')$.
\end{lemma}
\begin{proof}[Proof sketch]
	Note that repeated application of Lemma~\ref{l:cycle-grouping} allows us to group simple cycles together but it may create connecting paths of length greater than $|Q|$.
	In such a case, we extract simple cycles from long connecting paths and group them together again using Lemma~\ref{l:cycle-grouping}.
	Such process terminates as in each iteration it either
	decreases the sum of lengths' of connecting paths or does not change this sum but decreases
	the number of (non-grouped) cycles. Note that the initial cycle $\Path$ and the resulting cycle $\Path'$ have the same multisets of transitions and hence
	$\gain(\Path) = \gain(\Path')$ and $\select(\Path) = \select(\Path')$. It follows that this process preserves condition (C2).
\end{proof}

Consider a template $\template = (\connect_0, \cycle_1, \connect_1, \ldots, \cycle_p,\connect_p)$.
We present $\VAL_{\vec{g}}(\template(\vec{n}))$ as a function from multiplicities of simple cycles $\vec{n} \in \N$ into $\Z$.
First, observe that
\[
	\VAL_{\vec{g}}(\template(\vec{n})) = \VAL_{\vec{0}}(\template(\vec{n})) + \vec{g} \cdot \select(\template(\vec{n})).
\]
The expression $\select(\template(\vec{n}))$ is a linear expression in $\vec{n}$ with natural coefficients, and
the expression $\VAL_{\vec{0}}(\template(\vec{n}))$ is a quadratic function in each of its arguments $\vec{n}$.

\begin{restatable}{lemma}{TemplateValue}
	\label{l:template-value}
	Given a template $\template$ we can compute in polynomial time in $|\template|+|\vass|+|f|$, a symmetric matrix
	$\matrixB_{\template} \in \Z^{p\times p},  \vectorB_{\template} \in \Z^p$ and $\constB_{\template} \in \Z$ such that the following holds:
	\begin{equation}
		2\cdot \VAL_{\vec{0}}(\template(\vec{n})) = \vec{n}^T \matrixB_{\template} \vec{n} + \vectorB_{\template} \vec{n} + \constB_{\template}
		\label{cycle-value}
	\end{equation}
	Moreover, for all $i,j \in \{1, \ldots, p\}$ we have
	\begin{equation}
		\matrixB_{\template}[i,j] = \gain(\cycle_{min(i,j)}) \cdot \select(\cycle_{max(i,j)}).
		\label{eq:matrixBForm}
	\end{equation}
\end{restatable}
The proof of Lemma~\ref{l:template-value} has been relegated to the appendix.

Observe that $\matrixB_{\template}$ is similar to the matrix $\matrixA$ from~\eqref{eq:matrix-balanced}.
We exploit this similarity in the following section to eliminate the term $\vec{n}^T \matrixB_{\template} \vec{n}$ if possible.

\subsubsection{Elimination of the quadratic factor}

We show how to simplify the expression~\eqref{cycle-value} of Lemma~\ref{l:template-value} for $\VAL_{\vec{0}}(\template(\vec{n}))$.
We show that either the inequality
\[
	\VAL_{\vec{0}}(\template(\vec{n})) + \vec{g} \cdot \select(\template(\vec{n})) \leq 0
\]
has a simple solution for every $\vec{g}$, or
the quadratic term in~\eqref{cycle-value} of Lemma~\ref{l:template-value} can be substituted with a linear term.

\Paragraph{Negative and linear templates}.
Consider a template $\template$.
A template $\template$ is \emph{positive} (resp., \emph{negative}) if there exist multiplicities $\vec{n}_1, \vec{n}_2 \in \N^p$ such that
\begin{enumerate}[leftmargin=0.8cm,topsep=0pt,itemsep=-1ex,partopsep=1ex,parsep=1ex]
	\item $\vec{n}_1^T \matrixB_{\template} \vec{n}_1 > 0$ (resp., $\vec{n}_1^T \matrixB_{\template} \vec{n}_1 < 0$), and
	\item for every $t \in \N^+$, we have
	      \[
		      \gain( \template(t \vec{n}_1 + \vec{n}_2)) \cdot \select( \template(t \vec{n}_1 + \vec{n}_2)) =0.
	      \]
\end{enumerate}
A template $\template$ is \emph{linear} if there exist $\vectorC_{\template} \in \Z^p$ and $\constC_{\template} \in \Z$ such that
for all $\vec{n}$, if $\gain(\template(\vec{n})) \cdot \select(\template(\vec{n})) =0$,
then $\vec{n}^T \matrixB_{\template} \vec{n} = \vectorC_{\template} \cdot \vec{n} + \constC_{\template}$.
\medskip

We observe that the existence of a negative cycle $\template$ implies that $\VAL_{\vec{g}}(\template(\vec{n})) \leq 0$ has a solution for every $\vec{g} \in \N^k$, which in turn
implies that the answer to the average-value problem is YES. Basically, for
$\Path^t$ defined as $\template(t \vec{n}_1 + \vec{n}_2)$ and
$t$ big enough we can make
$\VAL_{\vec{g}}(\Path^t)$ arbitrarily small.

\begin{lemma}
	\label{l:negative-is-good}
	If there exist a negative template, whose any state  is reachable from some initial state,
	then the answer to the regular average-value problem with threshold $0$ is YES.
\end{lemma}
\begin{proof}[Proof sketch]
	Assume that $\template$ is negative and consider $\Path^t$ defined as $\template(t \vec{n}_1 + \vec{n}_2)$.
	Note that $2\cdot \VAL_{\vec{0}}(\Path^t) = (t \vec{n}_1 + \vec{n}_2)^T \matrixB_{\template}  (t \vec{n}_1 + \vec{n}_2) + \vectorB_{\template}  (t \vec{n}_1 + \vec{n}_2) + \constB_{\template}$ is
	a quadratic polynomial in $t$ with the leading coefficient $\vec{n}_1^T\matrixB_{\Path_1} \vec{n_1}$ negative, whereas $\vec{g} \cdot \select(\Path^t)$ is linear in $t$.
	Therefore, for $t$ big enough (which depends on $\vec{g}$) $\VAL_{\vec{g}}(\Path^t) = \VAL_{\vec{0}}(\Path^t) + \vec{g} \cdot \select(\Path^t)$ is negative and $\Path^t$ is balanced.
	It follows that $\Path_1$ leading from an initial state to the first state of $\template$ and $\Path_2 = \Path^t$ for $t$ big enough satisfy
	conditions (C1), (C2), (C3) of Lemma~\ref{l:conditions-avg} and hence the answer to the average-value problem is YES.
\end{proof}

We show that either there is a template, which is negative or all templates are linear (Lemma~\ref{l:negative-or-linear}).
Next, we show that we can check in non-deterministic polynomial time whether there exists a negative template (Lemma~\ref{l:linear-NP}).
We start with an example presenting one of the main ideas.

\begin{example}
	\label{ex:runningFour}
	Consider the VASS $\vass_e$ and the cost function $f$ from Example~\ref{ex:runningOne}.
	We consider balanced cycles $\Path$ starting in $B$. We know from Example~\ref{ex:runningThree} that in balanced cycles all transitions have the same number of occurrences.
	Therefore, $\Path$ is factorized into one of two minimal templates:
	$\template_1 = (\epsilon,e_1 e_2, \epsilon, e_3 e_4, \epsilon)$ or
	$\template_2 = (\epsilon,e_3 e_4, \epsilon, e_1 e_2, \epsilon)$.
	\[
		\matrixB_{\template_1} = \begin{bmatrix*}[r]
			4 & -1 \\
			-1 & 4
		\end{bmatrix*} \quad
		\matrixB_{\template_2} = \begin{bmatrix*}[r]
			4 & -7 \\
			-7 & 4
		\end{bmatrix*}
	\]
	Since transitions have the same number of occurrences, the multiplicities of $e_1 e_2$ and $e_3 e_4$ are equal and hence
	we consider $\vec{n}$ such that $n[1] = n[2]$. For such vectors $\vec{n}$ we have
	$\vec{n} \matrixB_{\template_1} \vec{n} = 6(n[1])^2$ and
	$\vec{n} \matrixB_{\template_2} \vec{n} = -6(n[1])^2$. Therefore, $\template_1$ is positive and $\template_2$ is negative.

	Interestingly, $\vec{n}^T \matrixB_{\template_1} \vec{n} = - \vec{n}^T \matrixB_{\template_2} \vec{n}$. This is not a coincidence.
	Recall the matrix $\matrixA'$ from Example~\ref{ex:runningThree} and observe that
	$\matrixA' = \matrixB_{\template_1} + \matrixB_{\template_2}$. As we consider
	$\vec{v}$ such that $\vec{n}^T \matrixA' \vec{n} = 0$ and hence the above equality follows.
\end{example}

\begin{restatable}{lemma}{Dichotomy}
	\label{l:negative-or-linear}
	(1)~There exists a negative template or all templates are linear.
	(2)~If there exists a positive template $\template$, then there exists a negative one of the size bounded by $|\template|^2$.
	(3)~If a template $\template$ is linear, we can compute $\vectorC_{\template}, \constC_{\template}$ in polynomial time in $|\template| + |\vass| + |f|$.
\end{restatable}
\begin{proof}[{Proof ideas}] Consider a template $\template$ with all connecting paths being empty, i.e., $\template = (\epsilon, \cycle_1, \ldots, \cycle_p, \epsilon)$.
	Let $\vec{n}$ be a vector of multiplicities such that a cycle $\template(\vec{n})$ is balanced, i.e.,
	$\gain(\template(\vec{n})) \cdot \select(\template(\vec{n})) = 0$. Then, $\template(t\vec{n})$ is balanced for all $t \in \N$.
	We consider three cases:

	\noindent\emph{The case $\vec{n}^T  \matrixB_{\template} \vec{n} < 0$}. Then, $\template$ is negative. Vectors $\vec{n}_1 = \vec{n}$ and $\vec{n}_2 = \vec{0}$ witness negativity.

	\noindent\emph{The case $\vec{n}^T  \matrixB_{\template} \vec{n} > 0$}. Then, $\template$ is positive and we show that the template with reversed cycles $\template^R = (\epsilon, \cycle_p, \ldots, \cycle_1, \epsilon)$ is negative.

	First, observe that $\gain(\template(\vec{n})) \cdot \select(\template(\vec{n})) = 0$ can be stated as a matrix equation $\vec{n}^T \matrixA \vec{n} =0$, where
	for all $i,j$ we have
	\[
		\matrixA[i,j] = \gain(\cycle_i) \cdot \select(\cycle_j) + \gain(\cycle_j) \cdot \select(\cycle_i).
	\]
	This can be derived in virtually the same way as in Section~\ref{sec:finite-average}.
	Next, recall from Lemma~\ref{l:template-value} that $\matrixB_{\template}$ is similar to $\matrixA$. For all $i,j$ we have
	\[
		\matrixB_{\template}[i,j] =  \gain(\cycle_{min(i,j)}) \cdot \select(\cycle_{max(i,j)})
	\]
	A similar identity holds for $\template^R$, but $i,j$ are substituted by respectively $p-i+1$ and $ p-j+1$ and $\min$ is swapped with $\max$. In consequence, we have
	\[
		\matrixB_{\template^R}[p+1-i,p+1-j] =  \gain(\cycle_{max(i,j)}) \cdot \select(\cycle_{min(i,j)})
	\]
	Therefore, for all $i,j$ we have
	\[
		\matrixA[i,j] = \matrixB_{\template}[i,j] + \matrixB_{\template^R}[p-i+1,p-j+1]
	\]
	It follows that
	\[
		0 = \vec{n}^T \matrixA \vec{n} =  \vec{n}^T \matrixB_{\template} \vec{n} +  \vec{n_R}^T \matrixB_{\template^R} \vec{n_R}
	\]
	where $\vec{n}_R$ is the reversed vector $\vec{n}$. Note that $\template^R(\vec{n}_R)$ has the same multiset of transitions as  $\template^R(\vec{n}_R)$, and hence
	$\template^R(\vec{n}_R)$ is balanced, and  $\template^R(t\vec{n}_R)$ is balanced for all  $t\in \N$.
	It follows that $\template^R$ is negative.
	\smallskip

	\noindent\emph{The above two cases fail}.
	Then for all $\vec{n}$, if $\gain(\template(\vec{n})) \cdot \select(\template(\vec{n})) = 0$, then
	$\vec{n}^T  \matrixB_{\template} \vec{n} = 0$. Therefore, $\template$ is linear with $\vectorC_{\template}, \constC_{\template}$ being $\vec{0}$ and $0$ respectively.
	In the general case, if the connecting paths are non-empty, then $\vectorC_{\template}, \constC_{\template}$ need not be $0$.
	Intuitively, $\gain(\template(\vec{n}))  \cdot  \select(\template(\vec{n}))$ depends on the whole template whereas $\matrixB_{\template}$ depends only on
	the cycles of $\template$. The expression $\vectorC_{\template} \cdot \vec{n} + \constC_{\template}$ accounts for that difference.
	\smallskip

	Essentially the same line of reasoning can be applied if the connecting paths in $\template$ are non-empty. However, the argument becomes more technical. We present it in full detail in the appendix.
\end{proof}

Now, we discuss how to check whether there exists a negative template.
One of the conditions of negativity is that
$\gain(\template(t \vec{n}_1 + \vec{n}_2)) \cdot \select(\template(t \vec{n}_1 + \vec{n}_2)) =0$.
Recall that $\gain$ and $\select$ depend on the multiset of transitions, but not on the order of transitions.
Therefore, for a given template $\template$,
we define $\trans_{\template}(\vec{n}) \in \N^m$, where $m  = |\delta|$, as the vector of multiplicities of transitions in $\template(\vec{n})$.
We will write $\trans(\vec{n})$ if $\template$ is clear from the context.

\newcommand{\vectorM}{\vec{z}}

\begin{restatable}{lemma}{ShortVectors}
	\label{l:short-vectors}
	Let $\template$ be a template and
	let $\vec{n} \in \N^p$ be a vector of multiplicities.
	There exist $r_1, \ldots, r_{\ell} \in \Q^+$ and $\vectorM_1, \ldots, \vectorM_{\ell} \in \N^p$
	such that
	\begin{enumerate}[label=(\arabic*)]
		\item $supp(\vectorM_i) \leq m$ (the number of transitions of $|\vass|$),
		\item there exists $t \in \N^+$ such that $\trans(\vectorM_i) = t \cdot \trans(\vec{n})$, and
		\item $\vec{n} = \sum_{i=1}^{\ell} r_i \vectorM_i$.
	\end{enumerate}
\end{restatable}
\noindent\emph{Remark}. The condition $\trans(\vectorM_i) = t \cdot \trans(\vec{n})$ implies that if the cycle $\template(\vec{n})$ is balanced, then
all the cycles $\template(\vectorM_1), \ldots, \template(\vectorM_{\ell})$ are balanced as well.
\smallskip
\begin{proof}[{Proof ideas}]
	First, observe that $\trans$ is a linear function transforming vectors from $\N^p$ into vectors from $\N^m$.
	The value $p$ can be exponential w.r.t. $m = |\delta|$ and hence we show that each vector from $\N^p$ can be presented as a  linear combination over $\Q^+$ of
	vectors with polynomially-bounded supports.

	The full proof has been relegated to the appendix.
\end{proof}

Next, we show that if there exists a negative template, then there exists one of polynomial size.

\begin{restatable}{lemma}{ShortNegativeTemplate}
	\label{l:short-negative-template}
	If there exists a negative template, then there exists one of polynomial size in $|\template| + |\vass| + |f|$.
\end{restatable}
\begin{proof}[{Proof ideas}]
	Assume that a template $\template$ is positive (resp., negative).
	Let $\vec{n}_1, \vec{n}_2 \in \N^p$ be the vectors witnessing positivity (resp., negativity).
	Then, using Lemma~\ref{l:short-vectors}, we show that there exist $\vec{n}_1^0, \vec{n}_2^0$ that witness template $\template$ being positive or negative and the support of both $\vec{n}_1^0, \vec{n}_2^0$ is bounded by $|\delta|$.
	We remove from $\template$ cycles corresponding to coefficient $0$ in both $\vec{n}_1^0$ and $\vec{n}_2^0$.
	We get a template with polynomially many cycles.
	If all connecting paths are still bounded by $|Q|$ we terminate with $\template$ of polynomial size.
	However, as we remove some cycles, some connecting path are concatenated and in the result we get connecting paths longer than $|Q|$.
	For such connecting path, we extract simple cycles and group them (as in Lemma~\ref{l:minimal-factorization}).
	We get another minimal template $\template'$ shorter than $\template$, which is positive or negative.
	By iterating this process, we get a polynomial-size template $\template^{\neq 0}$, which is positive or negative.
	Then, by Lemma~\ref{l:negative-or-linear} there exists a negative template of polynomial size.

	The full proof has been relegated to the appendix.
\end{proof}

Still, to check whether there exists a negative template we have to solve a system consisting of
a quadratic inequality $\vec{n}_1^T \matrixB_{\template} \vec{n}_1 < 0$ and
a quadratic equation, which corresponds to $\gain(\template(t\vec{n}_1 + \vec{n}_2)) \cdot \select(\template(t\vec{n}_1+ \vec{n}_2)) = 0$.
We show that this quadratic equation can be
transformed into a system of linear inequalities.

\begin{restatable}{lemma}{BalancedLinearInequalities}
	\label{l:balanced-inequalities}
	Let $\template$ be a template.
	There exist systems of linear equations and inequalities $S_1, \ldots, S_l$ such that
	(1)~each $S_i$ has polynomial size in $|\template| + |\vass| + |f|$,
	(2)~for all $\vec{n} \in \N^p$ we have $\gain(\template(\vec{n})) \cdot \select(\template(\vec{n})) = 0$ if and only if
	for some $i$ the vector $\vec{n}$ satisfies $S_i$.
\end{restatable}
\begin{proof}[{Proof ideas}]
	We construct a matrix $\matrixA \in \Z^{(p+1)\times (p+1)}$ such that
	for every $\vec{n} \in \N^p$ we have
	$\gain(\template(\vec{n})) \cdot \select(\template(\vec{n})) = 0$ if and only if for
	$\vec{n}_1 = (n[1], \ldots, n[p],1)$ we have
	$\vec{n}_1^T \matrixA \vec{n}_ = 0$.
	Moreover, we know that for all $\vec{n}_1 \in \N^{p+1}$
	we have  $\vec{n}_1^T \matrixA \vec{n}_1 \geq 0$.
	We carry out standard elimination of quadratic terms using the quadratic formula on the successive variables in $\vec{n}_1^T \matrixA \vec{n}_1$.
	Such a procedure involves computing squares in each step ($\Delta = b^2 -4ac$) and hence it can lead to coefficients of double-exponential order in $p+1$ (the number of dimensions).
	However, due to condition $\vec{n}_1^T \matrixA \vec{n_1} \geq 0$
	for all $\vec{n}_1 \in \N^{p+1}$, we can show that for every $i$
	either $x[i] = 0$ or $x[i]$ has to be a double root ($\Delta = 0$) and hence
	it can be described with a simple equation, which is of polynomial size.

	The full proof has been relegated to the appendix.
\end{proof}

Finally, Lemma~\ref{l:short-negative-template} and Lemma~\ref{l:balanced-inequalities} imply that
existence of a negative template can be solved in $\NP$ via reduction to integer quadratic programming.

\begin{lemma}
	\label{l:negativity-np}
	We can verify in $\NP$ whether a given $\VASS(\Z,k)$ $\vass$ has a negative template.
\end{lemma}

\begin{proof}[Proof sketch]
	We non-deterministically pick a template $\template$ of polynomial size (Lemma~\ref{l:short-negative-template}).
	Then, we non-deterministically pick a system $S_i$ for template (Lemma~\ref{l:balanced-inequalities}). We make two copies of $S_i$: $S_i^1$ and $S_i^2$; in one we
	substitute $\vec{n}$ with $\vec{n}_2$ and in the other with
	$\vec{n}_1 + \vec{n}_2$.
	It follows that for all $t \in \R$ the vector $t\vec{n}_1 + \vec{n}_2$ satisfies $S_i$
	(substituted for $\vec{n}$).
	Finally, we solve an instance of integer quadratic programming consisting of  $\vec{n}_1 \matrixB_{\template} \vec{n}_1 - 1 \leq 0$
	and linear equations and inequalities $S_i^1$ and $S_i^2$, which is an instance of integer quadratic programming and hence can be solved in $\NP$~\cite{del2017mixed}.
\end{proof}

\subsubsection{The linear case}

We consider the final case, where all templates are linear.
The decision procedure described in the following lemma answers YES
(in at least one of non-deterministic computations)
whenever the answer to the regular average-value problem with threshold $0$ is YES and
all templates are linear. Thus, it is complete.

Our main algorithm  assumes that all templates are linear if it fails to find a negative template.
The failure can be due to a wrong non-deterministic pick.
Having that in mind, we make sure that the decision procedure from the following lemma is sound regardless of the linearity of templates, i.e.,
if it answers YES, then the answer to the regular average-value problem with threshold $0$ is YES.

\begin{lemma}
	\label{l:linear-NP}
	Assume that all templates are linear.
	Then, we can solve the regular average-value problem with threshold $0$ in non-deterministic polynomial time.
	Moreover, the procedure is sound irrespectively of the linearity assumption.
\end{lemma}

\begin{proof}[Proof sketch]
	First, we show that using Lemma~\ref{l:short-vectors}, if there is a cycle $\Path$ with
	(a)~$\gain(\Path) \cdot \select(\Path) = 0$ and
	(b)~$\VAL_{\vec{g}}(\Path) \leq 0$, then there is a cycle $\Path'$ defined by a template of polynomial size
	satisfying both conditions (a) and (b).

	Assume that for some $\vec{g} \in \Z^k$ we have
	\begin{enumerate}
		\item $\VAL_{\vec{g}}(\Path) \leq 0$ and
		\item $\gain(\Path) \cdot \select(\Path) = 0$.
	\end{enumerate}
	Then, there exists a minimal template $\template$ and a vector of multiplicities $\vec{n}$ such that
	$\template(\vec{n}) = \Path$ (Lemma~\ref{l:minimal-factorization}).
	As $\template$ is linear, then we can compute in polynomial time $\vectorC_{\template} \in \Z^p, \constC_{\template} \in \Z$, which are polynomial in $|\template|+|\vass|+|f|$ such that
	\[
		2\cdot\VAL_{\vec{g}}(\template(\vec{n})) =  \big(\vectorB_{\template} +  \vectorC_{\template}\big) \cdot \vec{n} + \constB_{\template} + \constC_{\template} + 2\cdot \vec{g}\cdot \select(\template(\vec{n}))
	\]

	We apply Lemma~\ref{l:short-vectors} to vector $\vec{n}$ and get $\vec{n} = \sum_{i=1}^{\ell} r_i \vectorM_i$, with $\vectorM_i, r_i$ satisfying the statement of Lemma~\ref{l:short-vectors}.
	Observe that for some $i$,  we have $\VAL_{\vec{g}}(\template(\vectorM_i)) \leq 0$.
	As $\trans(\vectorM_i) = t \cdot \trans(\vec{n})$, we have
	\[
		\gain(\template(\vectorM_i)) \cdot \select(\template(\vectorM_i)) = t^2 \cdot \big(\gain(\template(\vec{n})) \cdot \select(\template(\vec{n})) \big) = 0
	\]
	Therefore, $\template(\vectorM_i)$ is a balanced cycle with  $\VAL_{\vec{g}}(\template(\vectorM_i)) \leq 0$.
	As $supp(\vectorM_i) \leq |\delta|$,  we remove from $\template$ cycles that correspond to $0$ components in $\vectorM_i$.
	Again, $\template$ may not be minimal so, as in the proof of Lemma~\ref{l:minimal-factorization}, we extract simple cycles and iterate the process until we get $\template'$ of polynomial size such that
	$\template'$ with some multiplicities defines a cycle $\Path'$, which satisfies $\VAL_{\vec{g}}(\Path') \leq 0$ and
	$\gain(\Path') \cdot \select(\Path') = 0$.

	Second, consider a template $\template$ of polynomial size.
	We nondeterministically pick a subset of states $Q$ and write a system of equations $S_\gain^Q$ over variables  $\vec{x}, \vec{y}$ such that
	$\vec{x}, \vec{y}$ is a solution of $S_{\gain}^Q$ if and only if there exists a path $\Path_1$ satisfying
	(a)~$\vec{x}$ are multiplicities of transitions along $\Path_1$
	(b)~$\Path_1$ is from some initial state of $\vass$ to the first state of $\template$,
	(c)~$\Path_1$ visits all states from $Q$,
	(d)~$\vec{y} = \gain(\Path_1)$.
	To write $S_{\gain}^Q$ we specify that all states from $Q$ have equal in-degree and out-degree except for the initial state of $\vass$ , in which out-degree is greater than in-degree by $1$, and
	the first state of $\template$, where in-degree is greater than out-degree by $1$.

	In summary, if all templates are linear, then there exist $\Path_1, \Path_2$ defining a regular computation of the average value at most $0$ if and only if the following holds:
	there is a subset of states $Q$ and a template $\template$ of polynomial size such that the system of inequalities consisting of
	$S^Q_{\gain}$ and  the inequality
	\begin{equation}
		\big(\vectorB_{\template,\vec{0}} +  \vectorC_{\template}\big) \cdot \vec{n}+ \constB_{\template,\vec{0}} + \constC_{\template} + 2\cdot \vec{y} \cdot \select(\template(\vec{n})) \leq 0
		\label{eq:linear-case}
	\end{equation}
	has a solution over natural numbers. Note that all the components except for $\vec{y} \cdot \select(\template(\vec{n}))$ are linear.
	Since $\vec{y}$ and $\vec{n}$ are variables, the component $\vec{y} \cdot \select(\template(\vec{n}))$ is quadratic.
	Still, $S^Q_{\gain}$ with $\eqref{eq:linear-case}$ is an instance of integer quadratic programming, which can be solved in $\NP$~\cite{del2017mixed}.

	We have assumed that $\template$ is linear. However, we cannot check that deterministically in polynomial time.
	Still, linearity of $\template$ is necessary only for the completeness of the above procedure.
	We can make sure that the answer is sound irrespectively of that assumption.
	Having a solution of the system, we can compute in polynomial time $\gain(\Path_1)$ and $\VAL_{\gain(\Path_1)}(\template(\vec{n}))$ and then verify
	$\VAL_{\gain(\Path_1)}(\template(\vec{n})) \leq 0$.
\end{proof}

\subsubsection{Summary}

We present a short summary of the non-deterministic procedure deciding whether a given $\VASS(\Z,k)$ $\vass$ has a regular computation of the value at most $0$.
We assume that all states in $\vass$ are reachable from initial states.

\begin{itemize}[leftmargin=10pt,itemindent=15pt]

	\item \emph{ Step 1}. We check whether there is a cycle $\Path$ with $\gain(\Path) \cdot \select(\Path) < 0$.
	      We can do it in non-deterministic polynomial time (Lemma~\ref{l:MatrixNegative}).
	      If the answer is YES, then the answer to the average-value problem is YES (Lemma~\ref{l:cycles}). Otherwise, we proceed.

	\item \emph{ Step 2}. We check whether there exists a negative template in $\vass$.
	      We can do it in non-deterministic polynomial time (Lemma~\ref{l:negativity-np}).
	      If the answer is YES, then the answer to the regular average-value problem is YES (Lemma~\ref{l:negative-is-good}). Otherwise, we proceed.

	\item \emph{ Step 3}. Assuming that the previous steps failed, all templates are linear.
	      Then, we can solve the regular average-value problem in non-deterministic polynomial time (Lemma~\ref{l:linear-NP}).
	      Note that Step 2, could have failed due to unfortunate non-deterministic pick.
	      However, the procedure from Lemma~\ref{l:linear-NP} is sound regardless of the linearity assumption.
\end{itemize}

In consequence, we have the main result of this section:

\begin{lemma}
	\label{l:average-in-np}
	The regular average-value problem for $\VASSes(\Z)$ with (general) cost functions is in $\NP$.
\end{lemma}

\subsection{Hardness}
\label{sec:np-hardness}

We show that the regular finite-value and the regular average-value problems are $\NP$-hard.
The proof is via reduction from the 3-SAT problem.
Given a 3-CNF formula $\varphi$ over $n$ variables,
we construct a VASS of dimension $2n$, where dimensions correspond to literals in $\varphi$.
Each simple cycle $\Path$ in the VASS consists of two parts:
The first part corresponds to picking a substitution $\sigma$, which is
stored in the vector $\gain(\Path)$.
The second part ensures that $\gain(\Path) \cdot \select(\Path) =0$ if $\sigma$ satisfies $\varphi$ and
it is strictly positive otherwise.
Therefore, if $\varphi$ is satisfiable, the VASS has a regular run of the average cost $0$, and otherwise all its regular runs
have infinite average cost. In  consequence, we have the following:

\begin{lemma}
	\label{l:np-hardness}
	The regular finite-value and the regular average-value problems for $\VASSes(\Z)$ with (general) cost functions are $\NP$-hard.
\end{lemma}
\begin{proof}
	Let $\varphi$ be a propositional formula in $3$-CNF with $n$ variables and $l$ clauses.
	We construct a $\VASS(\Z,2n)$ $\vass_{\varphi}$ such that (1)~if $\varphi$ is satisfiable, then
	$\vass_{\varphi}$ has a regular computation of value $0$, and (2)~otherwise, all regular computations of $\vass_{\varphi}$ have infinite value.

	The VASS $\vass_{\varphi}$, presented in Figure~\ref{fig:np-hardness}, consists of two parts.
	The first part consists of the states $q_0, \ldots, q_n$, where $q_0$ is initial in $\vass_{\varphi}$. It corresponds to picking a variable substitution.
	For every $i \in \{0, \ldots, n-1\}$, there are two transitions $(q_{i-1},q_{i},\OneVector{i})$ and $(q_{i-1},q_{i},\OneVector{i+n})$, i.e., transitions from $q_{i-1}$ to $q_i$ incrementing the $i$-th or $(i+n)$-th counter.
	Upon reaching $q_n$, the vector of counter values $\vec{g}$ has exactly $n$ components with value $1$ while other have value $0$.
	The vector $\vec{g}$ defines variable valuation $\sigma_{\vec{g}}$ defined as follows.
	A variable $p_j$ is valuated to true, if $g[j] = 1$, and it is false if $g[j+n] = 1$.
	The cost function $f$ in the first part is $0$ for all states.

	In the second part the values of counters do not change, but the cost function $f$ is non-zero.
	We use this part to ensure that the average value is $0$ if and only if  $\sigma_{\vec{g}}$ satisfies $\varphi$.
	This part consists of $l$ layers, each consisting of $3$ states.
	At the $j$-th layer, corresponding to the clause $c_j = l_1 \vee l_2 \vee l_3$, we define the cost function $f$ on states $q_{j,1}, q_{j,2}, q_{j,3}$ such that
	$f(q_{j,i}, \vec{g})$ returns the value of the component of $\vec{g}$ corresponding to $\neg l_i$.
	It follows that $\sigma_{\vec{g}}$ satisfies $c_j$ if and only if
	for one of these $3$ states we have $f(q_{j,i}, \vec{g}) = 0$.
	Therefore, a single cycle over $\vass$ has the average $0$ if and only if $\sigma_{\vec{g}}$ satisfies $\varphi$.

	It follows that if $\varphi$ has a satisfying valuation, then $\vass_{\varphi}$ has a regular computation of the average value $0$.
	Conversely, if  $\varphi$ is unsatisfiable, then for every simple cycle $\Path$ we have $\gain(\Path) \cdot \select(\Path) > 0$.
	Since all weights are non-negative, it follows that for all cycles $\Path$ we have $\gain(\Path) \cdot \select(\Path) > 0$ and hence
	all regular computations of $\vass_{\varphi}$ have the value $+\infty$.

	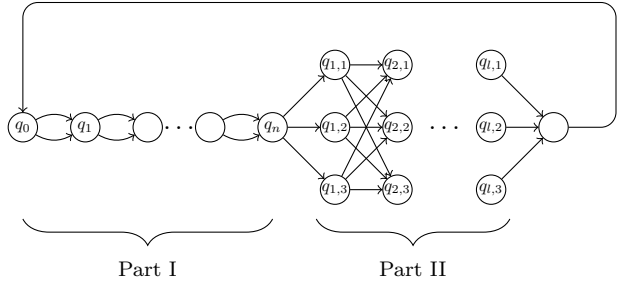
\begin{figure}
		\begin{center}
			\begin{tikzpicture}[scale=0.83]

				\tikzstyle{state}=[circle,draw, minimum size=0.6cm,inner sep=0cm,scale=0.65]
				\node[state] (Q0)  at (0,0) {$q_0$};
				\node[state] (Q1)  at (1,0) {$q_1$};
				\node[state] (Q2)  at (2,0) {};
				\node[state] (Qnn) at (3,0) {};
				\node[state] (Qn)  at (4,0) {$q_n$};

				\node at (2.5,0) {$\ldots$};

				\node[state] (c11) at (5,1) {$q_{1,1}$};
				\node[state] (c12) at (5,0) {$q_{1,2}$};
				\node[state] (c13) at (5,-1) {$q_{1,3}$};

				\node[state] (c21) at (6,1) {$q_{2,1}$};
				\node[state] (c22) at (6,0) {$q_{2,2}$};
				\node[state] (c23) at (6,-1) {$q_{2,3}$};

				\node[state] (cl1) at (7.5,1) {$q_{l,1}$};
				\node[state] (cl2) at (7.5,0) {$q_{l,2}$};
				\node[state] (cl3) at (7.5,-1) {$q_{l,3}$};

				\node at (6.75,0) {$\ldots$};

				\node[state] (final) at (8.5,0) {};

				\draw[->,bend left] (Q0) to (Q1);
				\draw[->,bend left] (Q1) to (Q2);
				\draw[->,bend left] (Qnn) to (Qn);

				\draw[->,bend right] (Q0) to (Q1);
				\draw[->,bend right] (Q1) to (Q2);
				\draw[->,bend right] (Qnn) to (Qn);

				\draw[->] (Qn) to (c11);
				\draw[->] (Qn) to (c12);
				\draw[->] (Qn) to (c13);

				\draw[->] (c11) to (c21);
				\draw[->] (c11) to (c22);
				\draw[->] (c11) to (c23);

				\draw[->] (c12) to (c21);
				\draw[->] (c12) to (c22);
				\draw[->] (c12) to (c23);

				\draw[->] (c13) to (c21);
				\draw[->] (c13) to (c22);
				\draw[->] (c13) to (c23);

				\draw[->] (cl1) to (final);
				\draw[->] (cl2) to (final);
				\draw[->] (cl3) to (final);

				\draw[->,rounded corners=0.2cm] (final) -- ++(1,0) -- ++(0,2) -- ++(-9.5,0) -- (Q0);

				\draw [decorate,decoration={brace,amplitude=10pt,mirror,raise=4pt},yshift=0pt]
				(0,-1.3) -- (4,-1.3) node [black,midway,yshift=-0.8cm] {\footnotesize
					Part I};

				\draw [decorate,decoration={brace,amplitude=10pt,mirror,raise=4pt},yshift=0pt]
				(4.7,-1.3) -- (7.8,-1.3) node [black,midway,yshift=-0.8cm] {\footnotesize
					Part II};

			\end{tikzpicture}
		\end{center}
		\caption{The structure of $\vass_{\varphi}$}
		\label{fig:np-hardness}
	\end{figure}

\end{proof}

The main results of this section (Lemma~\ref{l:MatrixNegative}, Lemma~\ref{l:average-in-np} and Lemma~\ref{l:np-hardness}) summarize to Theorem~\ref{th:IntegerMain}.
We leave the case of non-regular runs as an open question.

\begin{openquestion}
	What is the complexity of the average-value and finite-value problems for $\VASSes(\Z)$?
	\label{open:non-regular}
\end{openquestion}

\section{General cost functions and $\VASSes(\N)$}

We show that the average-value and the regular average-value problems are undecidable.

The proofs are via reduction from the halting problem for Minsky machines~\cite{Hopcroft:2006:IAT:1196416}, which are automata
with two natural-valued \emph{registers} $r_1, r_2$.
There are two main differences between Minsky machines and $\VASSes(\N)$.
First, the former can perform zero- and nonzero-tests on their registers, while the latter can take any transition as long as the counters' values remain non-negative.
Second, the halting problem for Minsky machines is qualitative, i.e.,  the answer is YES or NO.
We consider quantitative problems for VASS, where we are interested in the values assigned to computations.
We exploit the quantitative features of our problems to simulate zero- and nonzero-tests.

\begin{theorem}
	The average-value and the regular average-value problems for $\VASSes(\N)$ with (general) cost functions are undecidable.
\end{theorem}
\begin{proof}[Proof sketch]
	We first sketch the reduction from the halting problem for Minsky machines
	to the \emph{regular} average-value problem for $\VASSes(\N)$. To ease the presentation,
	we use \buchi{} conditions, which can be eliminated.

	Consider a Minsky machine $\M$ and let $I$ be its set of instructions.
	We construct a $\VASS(\N)$ $\vass$, which for presentation purposes has transitions labeled with the alphabet $I \cup \{ \# \}$, that satisfies the following.
	Every regular computation labeled by $v u^{\omega} \in (I \cup \{ \# \})^{\omega}$ has the average value $0$
	if and only if $u$ is of the form $u_1 \# \ldots u_k \#$
	and all words $u_1,  \ldots, u_k$ encode terminating computations of $\M$.
	First, we can construct an automaton without counters, which verifies that for every subword of the form $\# v \#$, with $v \in I^*$,
	$v$ encodes a terminating computations provided that all zero and non-zero test are correct.
	We can ensure that the cycle contains at least one transition labeled by $\#$ using a \buchi{} condition.
	This condition can be also enforced with additional counters.
	Such an automaton can be then encoded in the states of $\vass$.

	Second, to verify correctness of instructions w.r.t. the registers, we use counters $c_1, c_2$
	in $\vass$, which track the values of the registers minus $1$, i.e., the value of
	$c_i$ in $\vass$ is the maximum of $0$ and the value of the $i$-th register $r_i$ in $\M$.
	The VASS $\vass$ stores in its state whether the value of $r_i$ is $0$, $1$ or the value of $c_i$ plus $1$.
	With that information, the VASS $\vass$ can verify correctness of the instructions of $\M$.

	However, transitions of $\vass$ do not depend on the values of counters and hence it can be in a state denoting that
	the value of $r_i$ is $0$ or $1$, while the value of $c_i$ is greater than $0$.
	To ensure that this is not the case, the cost function $f$ returns the value of $c_i$ whenever $r_i$
	is supposed to be $0$ or $1$.
	This suffices to prove the correctness in the regular average-value case.

	Consider a regular computation $\comp$ of the average $0$ labeled with $v (u_1 \# \ldots u_k \#)^{\omega}$.
	The values returned by the cost function are natural and hence for the average to by $0$, the values at all positions need to be $0$.
	It follows that whenever $\vass$ stipulates that the value of $r_i$ is $0$ or $1$, then
	the value of $c_i$ is $0$. Thus, all zero- and non-zero tests in $u_1, \ldots, u_k$ are correct and hence these words encode terminating computations of $\M$. This concludes the reduction the the \emph{regular} average-value problem for $\VASSes(\N)$.

	In the average-value problem however, a computation of the value $0$ may be labeled with a word of the form $u_1 \# u_2 \# \ldots$, where the lengths of $u_i$'s grow fast.
	In that case we cannot infer that there exists $i$ such that all values of $f$ along $u_i$ are $0$. However, using two additional counters,
	we can ensure that all words $u_i$ have length bounded by $|u_0|$.
	The first counter computes the value $|u_1|-|u_2|+|u_3|- \ldots$, while the second computes $|u_2| - |u_3| + |u_4|  - \ldots$.
	Note that both counters are non-negative only if $|u_2| \leq |u_1|$, $|u_3| \leq |u_2|$, $|u_4| \leq |u_1| - |u_2| + |u_3|$.
	The first three inequalities imply $|u_4| \leq |u_1|$. We can show by induction that for all $i$ we have
	$|u_i| \leq |u_0|$.
	Therefore, the average-value is $0$ only if there exists (exists infinitely many) words $u_i$
	such that all values along this word are $0$. Then again, this word corresponds to a terminating computation of $\M$.
	Conversely, if $\M$ has a terminating computation and $u$ is its encoding, then $\vass$ has a computation of the average value $0$ that is labeled with $(\# u)^{\omega}$.
\end{proof}

Even though, the problems with an exact threshold are undecidable, we can decide existence of a regular computation of some finite value via reduction  to reachability in VASS.
\begin{theorem}
	The regular finite-value problem for $\VASSes(\N)$ with (general) cost functions is decidable.
\end{theorem}
\begin{proof}[Proof sketch]
	Note that a $\VASS(\N,k)$ $\vass$ has a regular computation of some finite value if and only if there exist finite paths $\Path_1, \Path_2$  such that
	$\Path_1 (\Path_2)^{\omega}$ corresponds to a computation and for all counters which influence the average in $\Path_2$,
	the values of this counters do not change in a full iteration over $\Path_2$.

	Assume that all counters influence $\Path_2$.
	We reduce the finite-value problem to the configuration reachability problem for $\VASS(\N, 2k)$.
	Such a VASS $\vass_{2k}$ first simulates $\vass$ maintaining two copies $r_{i,1}, r_{i,2}$ of each counter $c_i$. This
	phase of the $\vass_{2k}$ computation corresponds to $\Path_1$.
	Then, $\vass_{2k}$ nondeterministically picks a position, at it stores the current state of $\vass$ and switches to the mode
	where it modifies only counters $r_{i,1}$. This phase corresponds to $\Path_2$.
	Again, $\vass_{2k}$ non-deterministically picks the end of that phase and verifies that the last state is the same as the first of this phase.
	It remains to be verified that for all $i$ we have $r_{i,1} = r_{i,2}$. This is done by allowing to decrement $r_{i,1}, r_{i,2}$ simultaneously.
	Then, the final value is $\vec{0}$ if and only if for all $i$, at the end of the second phase we have $r_{i,1} = r_{i,2}$.
\end{proof}

The above argument relies on regularity of computations. We leave the general case as an open question.

\begin{openquestion}
	Is the finite-value problem for $\VASSes(\N)$ decidable?
	\label{open:finite-value}
\end{openquestion}

\section{VASS games}

The problems we studied correspond to finding a suitable path is a configuration-graph of a VASS, which can be considered as a single-player game on a VASS.
We briefly discuss a natural extension of our problems to two-player games on VASS, which are defined in a similar way to mean-payoff games on pushdown systems~\cite{CV17a}.
We consider infinite game arenas, in which the positions are the configurations of a VASS and moves correspond to transitions between configurations.
The states of the VASS are partitioned into Player-1 and Player-2 states and it
gives rise to the partition of configurations into Player-1 and Player-2 configurations.
On such an arena we define the \emph{cost} of each position as the cost of the configuration defined using the \emph{cost function}.
Finally, we consider two-player games with mean-payoff (cost) objectives;
As consequences of our results we have the following:
\begin{enumerate}
	\item As shown in Theorem~\ref{th:decidable-average-N}, configuration reachability reduces to the average-value problem, and for $\VASS(\N)$-games configuration reachability is undecidable~\cite{BrazdilJK10}.
	      Hence the average-value problem is undecidable for $\VASS(\N)$-games.
	\item Finite-state multidimensional total-payoff games straightforwardly reduce to the average-value problem for $\VASS(\Z)$-games; basically the counters track the total payoff and the general cost function specifies the multi-dimensional requirement.
	      Finite-state multidimensional total-payoff games are undecidable~\cite{Chatterjee0RR15}, and hence the average-value problem for $\VASS(\Z)$-games is undecidable.
\end{enumerate}

\section{Conclusion}
In this work we consider VASS with long-run average properties where the cost depends linearly on
the values of the counters.
For state-dependent and state-independent linear combinations
we establish decidability and complexity results for VASS with integer-valued and natural-valued counters.
We leave open some questions (see Open questions~\ref{open:natural}, \ref{open:non-regular}, and~\ref{open:finite-value}) which are interesting directions for future work.

\section*{Acknowledgment}
We would like to thank Petr Novotný for discussion on literature on VASS and the long-run average property.

\bibliographystyle{abbrv}
\bibliography{papers}

\clearpage

\section{Full proofs of selected lemmas}

We prove Lemma~\ref{l:cycles} as a consequence of the proof of Lemma~\ref{l:conditions-avg}, and hence we show this after the proof of Lemma~\ref{l:conditions-avg}.

\ConditionsAvg*
\begin{proof}
	Consider a regular computation $\comp$ let $\Path_1 (\Path_2)^\omega$ be the path corresponding to $\comp$.
	Assume that $\flimavg(f(\comp)) \leq 0$.
	Then, there is a sequence of positions $n$ such that $\frac{1}{n} \sum_{i=1}^n f(\comp[i])$ converges to $0$.
	Without loss of generality, we assume these positions are aligned with full occurrences of $\Path_2$, i.e.,
	$n = |\Path_1| + \ell \cdot |\Path_2|$.
	To see that, consider $r \in \{0,\ldots, |\Path_2|-1\}$ such that the infinite sequence of $n$'s
	has an infinite subsequence with the same remainder $r$ modulo $|\Path_2|$. Then, we can extend $\Path_1$ to $\Path_1'$ so that $|\Path_1'|$ modulo $|\Path_2|$ is $r$
	and rotate the cycle $\Path_2$.

	Therefore, we consider the sequence of partial averages
	\[
		\frac{1}{n} \sum_{i=1}^n f(\comp[i]) = \frac{1}{n} \big(\sum_{i=1}^{|\Path_1|} f(\comp[i]) +  \sum_{i=|\Path_1|+1}^{n} f(\comp[i])\big).
	\]
	where $n = |\Path_1| + \ell \cdot |\Path_2|$.
	First, observe that  as $n$ tends to infinity, $\frac{1}{n} (\sum_{i=1}^{|\Path_1|} f(\comp[i])$ converges to $0$.
	Therefore, $\frac{1}{n} \sum_{i=|\Path_1|+1}^{n} f(\comp[i])$ converges to $0$, and we study this term.

	Let $m = |\Path_2|$ and $\Path_2 = (q_1,q_2, \vec{y}_1) \ldots (q_m,q_1, \vec{y}_m)$.
	Observe that the value of the counters at the position $|\Path_1| + \ell \cdot|\Path_2| + i$ in $\comp$ equals
	\[
		\gain(\Path_1) + \ell \cdot\gain(\Path_2) + \sum_{j=1}^{i-1} \vec{y}_i,
	\]
	and the state at this position is $q_i$.
	Let $l : Q \to \N^k$ be the labeling defining $f$. Then, we have
	\[
		\sum_{i=1}^{m} f(\comp[|\Path_1|+(\ell m)+i]) =
		\sum_{i=1}^{m} \Big(\gain(\Path_1) + \ell \gain(\Path_2) +  \sum_{j=1}^{i-1} \vec{y}_j \Big)\cdot l(q_i)
	\]
	Note that
	\[
		\sum_{i=1}^{m} l(q_i) = \select(\Path_2)
	\]
	and hence
	\begin{equation}
		\sum_{i=1}^{m} f(\comp[|\Path_1|+(\ell m)+i]) = \gain(\Path_1) \cdot \select(\Path_2) +
		\ell \cdot \gain(\Path_2) \cdot \select(\Path_2) + \sum_{i=1}^{m} \Big(\sum_{j=1}^{i-1} \vec{y}_j \Big)\cdot l(q_i)
		\label{eq:parital-sum}
	\end{equation}

	Since $\gain(\Path_2) \cdot \select(\Path_2) = 0$ (Lemma~\ref{l:cycles}) we have
	\[
		\sum_{i=1}^{\ell m} f(\comp[|\Path_1|+i]) = \ell \VAL_{\gain(\Path_1)}(\Path_2).
	\]
	Since $\sum_{i=1}^{\ell m} f(\comp[|\Path_1|+i])$ tends to $0$ as $\ell \to \infty$ we have
	$\VAL_{\gain(\Path_1)}(\Path_2) \leq 0$.

	Therefore, $\flimavg(f(\comp)) \leq 0$ implies that $\VAL_{\gain(\Path_1)}(\Path_2) \leq 0$.
	Observe that the converse implication holds as well.
\end{proof}

As the consequence of the above proof we have the following:

\CycleCharacteristics*
\begin{proof}
	This follows directly from the equation \eqref{eq:parital-sum}:
	\begin{enumerate}
		\item if $\gain(\Path_2) \cdot \select(\Path_2) < 0$, then
		      $\lim_{\ell \to \infty} \sum_{i=1}^{m} f(\comp[|\Path_1|+(\ell m)+i]) = -\infty$ and hence
		      $\flimavg(f(\comp)) = -\infty$,
		\item if $\gain(\Path_2) \cdot \select(\Path_2) = 0$, then $\sum_{i=1}^{m} f(\comp[|\Path_1|+(\ell m)+i])$ is independent of $\ell$, and hence
		      $\flimavg(f(\comp)) = \sum_{i=1}^{m} f(\comp[|\Path_1|+i])$, which is finite, and
		\item if $\gain(\Path_2) \cdot \select(\Path_2) > 0$
		      $\lim_{\ell \to \infty} \sum_{i=1}^{m} f(\comp[|\Path_1|+(\ell m)+i]) = \infty$ and hence
		      $\flimavg(f(\comp)) = \infty$.
	\end{enumerate}
\end{proof}

\TemplateValue*
\begin{proof}
	First, observe that for all $\vec{g},\vec{h} \in \Z^k$ we have
	\begin{equation*}
		\begin{split}
			\VAL_{\vec{g}}(\PathTemp_1 \PathTemp_2) &=  \VAL_{\vec{g}}(\PathTemp_1)  + \VAL_{\vec{g}+\gain(\PathTemp_1)}(\PathTemp_2) \\
			\VAL_{\vec{g} + \vec{h}}(\PathTemp) &= \vec{g} \cdot \select(\PathTemp) + \VAL_{ \vec{h}}(\PathTemp).
		\end{split}
	\end{equation*}
	Therefore,
	\begin{equation*}
		\VAL_{\vec{0}}(\PathTemp_1 \PathTemp_2) = \VAL_{\vec{0}}(\PathTemp_1) + \VAL_{\vec{0}}(\PathTemp_2) + \gain(\Path_1) \cdot \select(\Path_2)
		\label{cycle-concat}
		\tag{$**$}
	\end{equation*}

	We repeatedly apply~\eqref{cycle-concat} to $\VAL_{\vec{0}}(\template(\vec{n}))$, where
	\[
		\template(\vec{n}) = \connect_0 \cycle_1^{n[1]} \connect_1 \cycle_2^{n[2]} \ldots \cycle_p^{n[p]} \connect_p
	\]
	and get
	\begin{equation}
		\begin{split}
			\VAL_{\vec{0}}&(\template(\vec{n})) = \VAL_{\vec{0}}(\connect_0) + \VAL_{\vec{0}}(\cycle_1^{n[1]}) + \ldots + \VAL_{\vec{0}}(\cycle_p^{n[p]})  + \VAL_{\vec{0}}(\connect_p) + \\
			+ \gain&(\connect_0)\cdot\select(\cycle_1^{n[1]}) + \gain(\connect_0 \cycle_1^{n[1]})\cdot\select(\connect_1) +\ldots +  \gain(\connect_0 \cycle_1^{n[1]} \ldots \cycle_p^{n[p]})\cdot\select(\connect_p)
		\end{split}
		\label{eq:first-template-identity}
	\end{equation}

	Observe that
	\begin{equation}
		\begin{split}
			\gain(\beta_i^{n[i]}) = n[i]&  \gain(\beta_i) \\
			\select(\beta_i^{n[i]}) = n[i]&  \select(\beta_i).
		\end{split}
		\label{eq:IdentitesTemplates}
	\end{equation}

	Therefore, by applying~\eqref{eq:IdentitesTemplates}
	to~\eqref{eq:first-template-identity} we get:
	\begin{equation}
		\begin{split}
			\label{eq:template-sum}
			2\VAL_{\vec{0}}(\beta_i^{n[i]}) = 2n[i] \VAL_{\vec{0}}(\beta_i) + n[i](n[i]-1) \gain(\beta_i) \cdot \select(\beta_i)
		\end{split}
	\end{equation}
	\begin{equation}
		\label{eq:template-connect}
		\begin{split}
			\gain(\connect_0 \cycle_1^{n[1]} \ldots \connect_{i-1}\cycle_i^{n[i]})\cdot\select(\connect_i) =
			\sum_{j=0}^{i-1} \gain(\connect_j)\select(\connect_i) + \sum_{j=1}^{i} n[j]\gain(\cycle_j)\select(\connect_i)
		\end{split}
	\end{equation}
	\begin{equation}
		\label{eq:template-cycle}
		\begin{split}
			\gain(\connect_0 \cycle_1^{n[1]} \ldots& \connect_{i})\cdot\select(\cycle_{i+1}^{n[i+1]}) = \\
			\sum_{j=0}^{i} n[i+1]\gain(\connect_j)\select(\cycle_{i+1})& +
			\sum_{j=1}^{i} n[j]n[i+1]\gain(\cycle_j)\select(\cycle_{i+1})
		\end{split}
	\end{equation}
	It follows that $2\VAL_{\vec{0}}(\template(\vec{n}))$ can be presented as a quadratic function of the form
	\begin{equation}
		2\cdot \VAL_{\vec{0}}(\template(\vec{n})) = \vec{n}^T \matrixB_{\template} \vec{n} + \vectorB_{\template} \vec{n} + \constB_{\template}
	\end{equation}
	Finally, the quadratic terms occur in \eqref{eq:template-sum} and \eqref{eq:template-cycle}.
	All quadratic terms in \eqref{eq:template-sum} are of the form
	$n[i]^2 \gain(\beta_i) \cdot \select(\beta_i)$ and all quadratic terms in \eqref{eq:template-cycle}
	are of the form
	$n[i]n[j] \gain(\beta_{j}) \cdot \select(\beta_{i})$, where $j<i$. Therefore, the matrix
	$\matrixB_{\template}$ is as in the statement.
\end{proof}

\Dichotomy*
\begin{proof}
	Consider a template $\template$ of the form
	\[
		\template = (\connect_0, \cycle_1, \ldots, \cycle_p, \connect_p)
	\]
	and a vector $\vec{n}$ such that
	$\template(\vec{n})$ is balanced ($\gain(\template(\vec{n})) \cdot \select(\template(\vec{n})) = 0$.)
	We define $\cycle_{p+1}$ as
	\[
		\cycle_{p+1}  = \template(\vec{0}) = \connect_0 \connect_1 \ldots \connect_p
	\]
	and we define $\vec{m} \in \Z^{p+1}$  as ${n}[i] = {m}[i]$ for $i \in \{1, \ldots p\}$ and $m[p+1] = 1$.
	We define an extended template $\template_1$ such that
	\[
		\template_1 = (\connect_0, \cycle_1, \ldots, \cycle_p, \connect_p,\cycle_{p+1},\epsilon).
	\]
	We distinguish three cases.
	\medskip

	\Paragraph{The case $\vec{m}^T  \matrixB_{\template_1} \vec{m} < 0$}.
	We show that $\template_1$ is negative.
	Consider vectors $\vec{n}_2, \vec{n}_2$ such that
	\[
		\begin{split}
			\vec{n}_1 &= (n[1], \ldots, n[p],1) = \vec{m}\\
			\vec{n}_2 &= (n[1], \ldots, n[p],0).
		\end{split}
	\]
	Then,  $\vec{n}_1^T  \matrixB_{\template_1} \vec{n_1} < 0$ and
	$\template_1(t\vec{n}_1 + \vec{n}_2)$ has the same multiset of edges as $\template(\vec{n})$ repeated $t+1$ times.
	Therefore,
	\[
		\begin{split}
			\gain(\template_1(t\vec{n}_1 + \vec{n}_2)) \cdot \select(\template_1(t\vec{n}+1 + \vec{n}_2)) = (t+1) \cdot \gain(\template(\vec{n})) \cdot \select(\template(\vec{n})) =0
		\end{split}
	\]
	and hence
	$\template_1$ is negative.
	\medskip

	\Paragraph{The case $\vec{m}^T  \matrixB_{\template_1} \vec{m} > 0$}.
	Similarly to the previous case, we can show that $\template_1$ is positive. However, we show that there exists a template $\template_2$, which is negative.
	To show that, we exhibit the connection between the equation $\gain(\template(\vec{n})) \cdot \select(\template(\vec{n})) = 0$ and $\vec{n}^T \matrixB_{\template_1} \vec{n}$.

	Since $\cycle_{p+1}  = \template(\vec{0})$ and $m[p+1] = 1$ we have
	\[
		\begin{split}
			\gain(\template(\vec{n})) &= \sum_{i=1}^{p+1} m[i] \cdot \gain(\cycle_i)\\
			\select(\template(\vec{n})) &= \sum_{i=1}^{p+1} m[i] \cdot \select(\cycle_i)
		\end{split}
	\]
	Therefore,
	\[
		\begin{split}
			\gain(\template(\vec{n})) \cdot \select(\template(\vec{n})) =
			\sum_{1\leq i,j \leq p+1} m[i][j] \gain(\cycle_i) \cdot \select(\cycle_j)
		\end{split}
	\]
	Observe that the last equation can be presented as a~quadratic form, i.e.,
	consider a matrix $\matrixC \in \Z^{(p+1) \times (p+1)}$ defined as for $i,j \in \{1,\ldots, p+1\}$ as
	\[
		\matrixC[i,j] = \gain(\cycle_i)\cdot\select(\cycle_j) + \gain(\cycle_j)\cdot\select(\cycle_i).
	\]
	Then,
	\[
		\gain(\template(\vec{n})) \cdot \select(\template(\vec{n})) = \vec{m}^T \matrixC \vec{m}  = 0.
	\]
	Recall, that for all $1 \leq i,j \leq p+1$ we have
	\[
		\matrixB_{\template_1}[i,j] =  \gain(\cycle_{min(i,j)}) \cdot \select(\cycle_{max(i,j)})
	\]

	Consider any template $\template_2$ with the same cycles as $\template_1$ but in the reverse order.
	Then, for all $1 \leq i,j \leq p+1$ we have
	\[
		\begin{split}
			\matrixB_{\template_2}[i,j] =  &\gain(\cycle_{max((p+2)-i,(p+2)- j)}) \cdot \\
			&\select(\cycle_{min((p+2)-i,(p+2)- j)})
		\end{split}
	\]
	and hence
	\[
		\matrixB_{\template_2}[(p+2)-i,(p+2)-j] = \gain(\cycle_{max(i,j)}) \cdot \select(\cycle_{min(i,j)}).
	\]
	It follows that $\matrixB_{\template_2}$ equals $\matrixC - \matrixB_{\template_1}$ transposed along anti-diagonal, i.e.,
	for all $1 \leq i,j \leq p+1$ we have
	\[
		(\matrixC - \matrixB_{\template_1})[i,j] = \matrixB_{\template_2}[(p+2)-i,(p+2)-j]
	\]

	Therefore, $\vec{m}^T  (\matrixC - \matrixB_{\template_1})  \vec{m}$, which is less than $0$, equals $\vec{m}_R^T  \matrixB_{\template_2} \vec{m}_R$, where
	$\vec{m}_R$ is the reversed vector $\vec{m}$.

	It remains to show that there exists such $\template_2$ of the size bounded by $|\template_1|^2$.
	We define $\template_2$ as $\cycle_{p+1}$ followed by $p$ iterations of $\template(\vec{0}) = \connect_0 \ldots \connect_p$ each with a signle occrrence of a cycle $\cycle_i$ in the order
	$\cycle_{p}, \cycle_{p}, \ldots, \cycle_1$. Formally, for $1 \leq i \leq j \leq p+1$ we define
	\[
		\connect[i,j] = \connect_i \connect_{i+1} \ldots \connect_{j}
	\]
	We define $\template_2$ as follows:
	\[
		\begin{split}
			\template_2 := (\epsilon,\cycle_{p+1},\connect[0,p-1],\cycle_{p},\connect[p]\connect[0,p-2] ,\ldots,\connect[2,p+1]\connect[1],\cycle_1,\connect[1,p+1]),
		\end{split}
	\]
	Observe that the path $\template_2(\vec{0}) = \cycle_{p+1}^{p}$, i.e., it is the same as $\template_1(\vec{0})^{p} = \template(\vec{0})^{p}$.
	Therefore, the size of $\template_2$ is bounded by $|\template_1|^2$.
	The template $\template_2$ is not minimal, but it has not been a requirement.

	The vectors
	\[
		\begin{split}
			\vec{k}_1 = p\vec{m}_R =  &(p,pn[p], \ldots, pn[1]) , \textrm{ and} \\
			\vec{k}_2 =  &(0,pn[p], \ldots, pn[1])
		\end{split}
	\]
	witness $\template_2$ being negative.
	Indeed, first
	\[
		(\vec{k}_1)^T  \matrixB_{\template_1} (\vec{k}_1)  = p^2 \vec{m}_R^T  \matrixB_{\template_2} \vec{m}_R < 0.
	\]
	Second, $\template_2(t\vec{k}_1 + \vec{k}_2)$ has the same multiset of edges as $\template(\vec{n})$ repeated
	$p(t+1)$ times.
	Therefore, as in the previous case, we have
	\[
		\gain(\template_2(t\vec{k}_1 + \vec{k}_2)) \cdot \select(\template_2(t\vec{k}_1 + \vec{k}_2)) = 0.
	\]
	It follows that $\template_2$ is negative.
	\medskip

	\Paragraph{The case $\vec{m}^T  \matrixB_{\template_1} \vec{m} = 0$}.
	We show that $\template$ is linear.
	Observe that $\matrixB_{\template}$ is obtained from $\matrixB_{\template_1}$ by deletion of the $(p+1)$-th row and the $(p+1)$-th column.
	Therefore,
	\[
		\begin{split}
			\vec{n}^T \matrixB_{\template} \vec{n} = \vec{m}^T  \matrixB_{\template_1} \vec{m}
			- \sum_{i=1}^p \gain(\cycle_i) \cdot \select(\cycle_{p+1}) -\\
			\sum_{i=1}^p \gain(\cycle_{p+1}) \cdot \select(\cycle_{i})
			- \gain(\cycle_{p+1}) \cdot \select(\cycle_{p+1}).
		\end{split}
	\]
	Therefore, we define $\vectorC_{\template}$ for $1 \leq i \leq p$ as
	\[
		\vectorC_{\template}[i] = - \big(\gain(\cycle_i) \cdot \select(\cycle_{p+1}) + \gain(\cycle_{p+1}) \cdot \select(\cycle_{i}) \big),
	\]
	and we define $\constC_{\template} = - \gain(\cycle_{p+1}) \cdot \select(\cycle_{p+1})$.
	Since $\vec{m}^T  \matrixB_{\template_1} \vec{m} = 0$, we get
	\[
		\vec{n}^T \matrixB_{\template} \vec{n} = \vectorC_{\template} \cdot \vec{n} + \constC_{\template}.
	\]
	Therefore, $\template$ is linear.

	Finally, if no template is negative, then all templates are linear.
\end{proof}

The following lemma is a vital part of the proof of Lemma~\ref{l:short-vectors}.
Recall that for a given template $\template$,
we define $\trans_{\template}(\vec{n}) \in \N^m$, where $m  = |\delta|$,
as the vector of multiplicities of transitions in $\template(\vec{n})$.
The function $\trans_{\template}$ is linear.
We write $\trans(\vec{n})$ if $\template$ is clear from the context.

\begin{lemma}
	\label{l:subvector}
	Let $\template$ be a template.
	For every vector $\vec{n} \in \N^p$, there exists $\vec{n}_0 \in \N^{p}$ such that
	\begin{enumerate}[label=(\arabic*)]
		\item $supp(\vec{n}_0) \leq m$,
		\item $supp(\vec{n}_0) \subseteq supp(\vec{n})$, and
		\item there exists $t \in \N^+$ such that $\trans(\vec{n}_0) = t \trans(\vec{n})$.
	\end{enumerate}
\end{lemma}
\begin{proof}
	We define $\trans^* \colon \Q^p \to \Q^m$ as the unique linear extension of $\trans$.
	Assume that $supp(\vec{n}) > m$. Then, there are $m+1$ components $i_1, \ldots, i_{m+1}$ of $\vec{n}$
	such that  $n[i_1],\ldots,n[i_{m+1}]$ are non-zero and
	$\trans^*(\OneVector{i_1}), \ldots, \trans^*(\OneVector{i_{m+1}})$ are linearly dependent.
	It follows that  $\{i_{1}, \ldots, i_{m+1} \}$ can be partitioned into $I_1, I_2$ and there exist
	non-zero vectors $\vec{y}_1,\vec{y}_2 \in \Q^p$ with the support included in $I_1$ and $I_2$ respectively.
	such that $\trans^*(\vec{y}_1) = \trans^*(\vec{y}_2)$.
	Consider the maximal $c \in \Q^+$ such that all components of $\vec{x} = \vec{n} - c \vec{y}_1 + c \vec{y}_2$ are
	non-negative, i.e., $\vec{x} \in (\Q^{\geq 0})^{p}$.
	Observe that $\trans^{*}(\vec{x}) = \trans^{*}(\vec{n})$, $supp(\vec{x}) \subseteq supp(\vec{n})$ and
	at least one component of $\vec{n}$ becomes $0$ in $\vec{x}$.

	We iterate this process until we get a vector $\vec{x}_0 \in (\Q^{\geq 0})^{p}$
	such that
	\begin{enumerate}[label=(\arabic*)]
		\item $supp(\vec{x}_0) \leq m$,
		\item $supp(\vec{x}_0) \subseteq supp(\vec{n})$, and
		\item $\trans^{*}(\vec{x}_0) = \trans^{*}(\vec{n})$.
	\end{enumerate}
	Finally, let $t \in \N$ be such that $t \vec{x}_0 \in \N^p$.
	Observe that $\vec{n}_0 = t \vec{x}_0$ satisfies the statement.
\end{proof}

\ShortVectors*
\begin{proof}
	We show the lemma by induction on $|supp(\vec{n})|$.
	Clearly, if $|supp(\vec{n})| \leq m$ the statement holds.

	Consider $\vec{n} \in \N^p$ and assume that for all $\vec{n}' \in \N^p$ with $|supp(\vec{n}')| < |supp(\vec{n})|$ the lemma holds.
	Let $\vec{n}_0$ be the vector from Lemma~\ref{l:subvector}.
	Consider the maximal $r \in \Q^+$ such that $\vec{n} - r \vec{n}_0$ has all the components non-negative.
	Let $r = \frac{p}{q}$ and consider $\vec{k} = q\vec{n} - p \vec{n}_0$.
	Observe that $\vec{k} \in \N^p$ and $|supp(\vec{k})| < |supp(\vec{n})|$ and hence
	there exist $s_1, \ldots, s_{\ell} \in \Q^+$ and $\vec{y}_1, \ldots, \vec{y}_{\ell} \in \N^p$
	such that for all $1 \leq i \leq \ell$
	\begin{enumerate}[label=(\arabic*)]
		\item $supp(\vec{y}_i) \leq m$ (the number of transitions of $|\vass|$),
		\item there exist $t \in \N^+$ such that $\trans(\vec{y}_i) = t \cdot \trans(\vec{k})$, and
		\item $\vec{k} = \sum_{j=1}^{\ell} s_j \vec{y}_j$.
	\end{enumerate}

	Recall that there is $t_0 \in \N^+$ such that $\trans(\vec{n}_0) = t_0 \trans(\vec{n})$. It follows that
	\[
		\trans(\vec{k}) = (q - t_0) \trans(\vec{n})
	\]
	Therefore,
	\[
		r_1 = \frac{1}{q}s_1, \ldots, r_{\ell} = \frac{1}{q}s_{\ell}, r_{\ell+1} = \frac{p}{q}
	\]
	and
	\[
		\vectorM_1 = \vec{y}_1, \ldots, \vectorM_{\ell} = \vec{y}_{\ell}, \vec{n}_0
	\]
	satisfy the statement of this lemma.
\end{proof}

\ShortNegativeTemplate*
\begin{proof}
	Assume that a template $\template_1$ is positive (resp., negative).
	First, we show that positivity (resp., negativity) of $\template_1$ is witnessed by vectors $\vec{k}_1, \vec{k}_2$ with at most $m$ non-zero components.

	Let $\vec{n}_1, \vec{n}_2 \in \N^p$ be the vectors witnessing positivity (resp., negativity).
	Let $r_1, \ldots, r_{\ell} \in \Q^+$ and $\vectorM_1, \ldots, \vectorM_{\ell}$ be coefficients and vectors for $\vec{n}_1$ satisfying
	Lemma~\ref{l:short-vectors}. Then, either for some $i$ we have $\vectorM_i \matrixB_{\template} \vectorM_i \neq 0$,
	and we define $\vec{p}_1$ as $\vectorM_1^i$ or for some pair $i \neq j$
	\[
		(r_i\vectorM_i + r_j\vectorM_j) \matrixB_{\template} (r_i\vectorM_i + r_j\vectorM_j) \neq 0
	\] and we define
	$\vec{p}_1 = a(r_i\vectorM_i + r_j\vectorM_j)$, where $a$ is a natural number such that $ar_i, ar_j \in \N$.
	Note that one of these conditions holds; if they both fail, then
	due to $\vec{n}_1 = \sum_{i=1}^{\ell} r_i \vectorM_i$, we have
	$\vec{n}_1 \matrixB_{\template} \vec{n}_1 = 0$ contrary to the assumptions.

	Let $\vec{p}_2$ be a vector with $supp(\vec{m}_2) \leq |\delta|$ and
	$\trans(\vec{p}_2) = t\trans(\vec{n}_2)$ (which exists by Lemma~\ref{l:short-vectors}).
	Finally, we define  $\vec{k}_1 = a \vec{p}_1$
	and $\vec{k}_1 = b \vec{p}_2$ for $a,b \in \N^+$ such that for some $t \in \N$ we have
	\begin{align*}
		\trans(\vec{k}_1) & = t\trans(\vec{n}_1) \\
		\trans(\vec{k}_2) & = t\trans(\vec{n}_2)
	\end{align*}
	Observe that $\vec{k}_1, \vec{k}_2$ also witness template $\template_1$ being positive or negative, but
	$|supp(\vec{k}_1)|,|supp(\vec{k}_2)| \leq m$.

	We remove from $\template_1$ cycles corresponding to coefficient $0$ in both $\vec{k}_1$ and $\vec{k}_2$.
	We get a template with polynomially many cycles. If all connecting paths are still bounded by $|Q|$ we terminate with $\template_1$, of polynomial size.
	However, as we remove some cycles, some connecting path are concatenated and in the result we get connecting paths longer than $|Q|$.
	For such connecting path, we extract simple cycles and group them (as in Lemma~\ref{l:minimal-factorization}).
	We get another minimal template $\template_2$ shorter than $\template_1$, which is positive or negative.
	By iterating this process, we get a polynomial-size template $\template^{\neq}$, which is positive or negative.
	Then, by Lemma~\ref{l:negative-or-linear} that there exists a negative template of polynomial size.
\end{proof}

\BalancedLinearInequalities*
\begin{proof}
	Recall the construction of $\template_1$ from the proof of Lemma~\ref{l:negative-or-linear}.
	We repeat this construction and define a symmetric matrix $\matrixA \in \Z^{(p+1) \times (p+1)}$
	such that for all $\vec{n} \in \N^p$ we have
	\[
		2 \gain(\template(\vec{n})) \cdot \select(\template(\vec{n})) = \vec{n}_1^T \matrixA \vec{n}_1
	\]
	where $\vec{n}_1 = (n[1], n[2], \ldots, n[p],1)$.
	We define $\cycle_{p+1} = \connect_0 \connect_1 \ldots \connect_p$ and
	the $\matrixA$ is defined for all $1 \leq i,j \leq p+1$ as
	\[
		\matrixA[i,j] = \gain(\cycle_{i})\cdot\select(\cycle_j) + \gain(\cycle_{j})\cdot\select(\cycle_i).
	\]

	We have assumed that $\vass$ does not have a regular run of the value $-\infty$ and hence
	for all $\vec{m} \in \N^{p+1}$ we have $\vec{m}^T \matrixA \vec{m} \geq 0$.
	We eliminate variables using the standard quadratic formula.
	We start with $m[1]$ and observe that
	\[
		\vec{m}^T \matrixA \vec{m} = A_1 (m[1])^2 + B_1 m[1] + C_1
	\]
	where $A_1 = \matrixA[1,1]$,
	$B_1 = 2\sum_{i=2}^{p+1} \matrixA[1,i]$, and
	$C_1 = \vec{n}_0^T \matrixA \vec{n}_0$ where $\vec{m}_0 = (0,m[2], m[3], \ldots, m[p+1])$.

	We compute $\Delta = (B_1)^2 - 4A_1C_1$ and assume that $\Delta \geq 0$.
	Observe that there are four cases:

	\begin{itemize}[leftmargin=10pt,itemindent=15pt]
		\item \textbf{The case: $A_1 = B_1 = 0$}. We observe that this is possible only if for all $i$ we have $\matrixA[1,i] = \matrixA[i,1] = 0$.

		      Suppose that there is $i$ such that $\matrixA[1,i] \neq 0$. Then, since $B_1 = 0$ there is $j$ such that
		      $\matrixA[1,j] \leq -1$. Consider $\vec{m}_0 \in \N^{p+1}$ such that $m_0[1] = \matrixA[1,j]^2$, $m_0[j] = 1$ and all the other components are $0$.
		      Then,
		      $\vec{m}_0^T \matrixA \vec{m}_0 < 0$, which contradicts the assumptions.

		\item \textbf{The case: $A_1 = 0$ and $B_1 \neq 0$}. In this case, we argue that $m[1] = 0$.
		      Indeed, we have $C_1 \geq 0$. Now, if $B_1 < 0$, then for $\vec{m}_t = (t,m[2],\ldots, m[p+1])$
		      and $t\in\N^+$ big enough, we get $\vec{m}_t^T \matrixA \vec{m}_t < 0$ contrary to the assumptions.
		      Therefore, $B_1 > 0$. It follows that $ B_1 m[1] + C_1 = 0$ if and only if $m[1] = 0$ and $C_1 = 0$.

		\item \textbf{The case $A_1 \neq 0$ and $\Delta > 0$}. Then, we argue that $m[1] = 0$.
		      Indeed, if $\vec{x} \in (\R^{\geq 0})^{p+1}$ is such that $x[1] > 0$, then for
		      some small $\epsilon \in (-\frac{x[1]}{2},\frac{x[1]}{2})$ and the vector $\vec{x}_1 = (m[1]+\epsilon, m[2], \ldots, m[p+1])$ we have $\vec{x}_1^T \matrixA \vec{x}_1 < 0$.
		      It follows that there exists $\vec{y} \in (\Q^{\geq 0})^{p+1}$ close to $\vec{x}_1$ such that
		      $\vec{y}^T \matrixA \vec{y} < 0$. Finally, for some $t \in \N$ we have $t \vec{y} \in \N^{p+1}$
		      and
		      \[
			      (t\vec{y})^T \matrixA (t\vec{y}) = t^2 \big(\vec{y}^T \matrixA \vec{y} \big) < 0
		      \]
		      which contradicts the assumption $\vec{m}^T \matrixA \vec{m} \geq 0$.

		\item \textbf{The case: $A \neq 0$ and $\Delta = 0$}. Then, $m[1] = \frac{-B_1}{2A_1}$.
	\end{itemize}

	We have the same identities for all components of $\vec{m}$. Now, we construct
	linear inequalities $S_i$ for $\vec{n}$.

	First, without loss of generality we assume that $\matrixA$ does not have $\vec{0}$ rows or columns. Such rows and columns in the symmetric $\matrixA$ correspond to variables, which occur only with the $0$ coefficient.
	Second, recall that we assumed that for $1 \leq i \leq p$ we have $m[i] = n[i]$ and $m[p+1] = 1$. We substitute $m[p+1]$ with $1$ in all the equations on $\vec{n}$.
	Third, for every subset $P \subseteq \{1,\ldots, p\}$, we define $S_P$ such that $P = \{ i \colon n[i] = 0\}$.

	Consider $P \subseteq \{1,\ldots, p\}$. We define $S_P$ as
	\[
		\begin{split}
			&\{ n[i] = 0 \colon i \in P\} \cup \\
			&\{ \matrixA[i,i]n[i] = -\big(\sum_{j\neq i} \matrixA[i,j]n[j]\big) - \matrixA[i,p+1] \colon i \notin P\}
		\end{split}
	\]
	We assume that $i,j \in \{1, \ldots, p\}$ which is omitted for readability.

	First, note that for every solution $\vec{n} \in \N$ to $S_P$ we have
	$\vec{n}_1 = (n[1], n[2], \ldots, n[p],1)$ satisfies $\vec{n}_1^T \matrixA \vec{n}_1 = 0$.
	We can remove all rows and columns with indexes from $P$ as their contribution to $\vec{n}_1^T \matrixA \vec{n}_1$ is $0$.
	Let $\matrixA'$ be the matrix and $\vec{k}$ be the vector resulting from removal of components from  $P$.
	Then, we observe that $\matrixA' \vec{k}$ is the $\vec{0}$ vector. Indeed, for every $i$ the equation
	\[
		\matrixA[i,i]n[i] = \sum_{j\neq i} \matrixA[i,j]n[j] + \matrixA[i,p+1]
	\]
	implies that the $i$-th component of $\matrixA \vec{n}_1$ is $0$. Thus, all components of the minor $\matrixA'$ with $\vec{k}$ are $0$.

	Finally, we argue that for every $\vec{m}$ satisfying $\vec{m}^T \matrixA \vec{m} = 0$
	with $m[p+1] = 1$, the vector $\vec{n} = (m[1], \ldots, m[p])$ solves some $S_P$.
	Indeed, we consider $P = \{ i \colon n[i] = 0\}$ and observe that due to the above considerations for all
	$i \in \{1,\ldots,p\} \setminus P$ we need to have $2A_i n[i] = {-B_i}$ as it is stated in $S_P$.
\end{proof}

\end{document}